\documentclass[letterpaper]{article} 

\usepackage{aaai2027}

\usepackage[hyphens]{url}  % DO NOT CHANGE THIS
\usepackage{graphicx} % DO NOT CHANGE THIS
\usepackage{natbib}  % DO NOT CHANGE THIS AND DO NOT ADD ANY OPTIONS TO IT
\usepackage{caption} % DO NOT CHANGE THIS AND DO NOT ADD ANY OPTIONS TO IT
\usepackage{algorithm}
\usepackage{algorithmic}
\usepackage{amsthm}
\usepackage{mathtools}
\usepackage{multirow}
\usepackage{amssymb}
\usepackage{pifont}

\usepackage{float}

\usepackage[ruled,linesnumbered,algo2e]{algorithm2e}

\SetKwInput{KwParam}{Parameters}

\SetAlCapFnt{\bfseries}
\SetAlCapNameFnt{\bfseries}
\SetNlSty{textbf}{}{}

\nocopyright

\newtheorem{prop}{Proposition}
\newtheorem{lemma}{Lemma}
\newtheorem{thm}{Theorem}

\newtheorem{assumption}{Assumption}
\newtheorem{definition}{Definition}

\newcommand{\cmark}{\ding{51}} % ✓
\newcommand{\xmarkgray}{{\color{lightgray}\ding{55}}} % ✗

\usepackage{newfloat}
\usepackage{listings}
\DeclareCaptionStyle{ruled}{labelfont=normalfont,labelsep=colon,strut=off} % DO NOT CHANGE THIS
\floatstyle{ruled}
\newfloat{listing}{tb}{lst}{}
\floatname{listing}{Listing}

\usepackage{booktabs}

\title{Barycentric Fused Gromov-Wasserstein Balancing \\ for Causal Inference under Multiple Treatments}
\author{
    Yuki Murakami\textsuperscript{\rm 1},
    Takumi Hattori\textsuperscript{\rm 1},
    Kohsuke Kubota\textsuperscript{\rm 2},
}
\affiliations{
    \textsuperscript{\rm 1} NTT DOCOMO, INC., Tokyo, Japan\\
    \textsuperscript{\rm 2} Yokohama City University, Kanagawa, Japan\\

    yuuki.murakami.xg@nttdocomo.com,
    takumi.hattori.zw@nttdocomo.com,
    kubota.kos.oy@yokohama-cu.ac.jp

}

\begin{document}

\maketitle

\begin{abstract}

Estimating heterogeneous single and interaction treatment effects from observational data under multiple simultaneous treatments is crucial for decision-making. 
To mitigate estimation variance, previous studies balance representation distributions between every pair of treatment patterns. 
However, such pairwise balancing scales quadratically with the number of treatment patterns and fails to preserve consistent local proximity structures across patterns, which degrades counterfactual estimation. 
To address these challenges, we propose the Causal Inference for Heterogeneous Single and Interaction Treatment Effects Network~(CIHSI-Net), a deep learning framework built on a novel Barycentric Fused Gromov-Wasserstein Balancing~(BFG-WB) objective. 
BFG-WB aligns the representation distribution of each treatment pattern with a shared Wasserstein barycenter, achieving global alignment while reducing the computational complexity from quadratic to linear, and its Fused Gromov-Wasserstein discrepancy preserves the local proximity structures essential for reliable heterogeneous effect estimation. 
Simulation studies show that CIHSI-Net consistently outperforms state-of-the-art baselines, and an application to real-world marketing data demonstrates its practical utility in complex multi-treatment scenarios.
\end{abstract}

\section{Introduction}
Estimating heterogeneous single and interaction treatment effects under multiple simultaneous treatments is important in many domains, such as healthcare and marketing.
For example, optimizing combination drug therapies requires understanding how these effects vary due to physiological differences~\cite[e.g.][]{gradman2010combination,sever2006potential}.
Similarly, in marketing, promotion effectiveness depends heavily on user attributes and their interactions~\cite[e.g.][]{blake2015consumer,lesscher2021offline}.
Given that randomized trials are often infeasible due to ethical concerns and the high costs associated with the combinatorial explosion of treatment options, analytical frameworks that can accurately estimate these effects from observational data are essential for decision-making in real-world applications.

To enable such accurate estimation from observational data, representation balancing with deep learning has emerged as a promising approach to address the inherent selection bias.
This approach mitigates distributional discrepancies between treatment and control groups by learning a balanced representation~\cite[e.g.][]{shalit2017estimating, yao2018representation}.
Recently, Causal Inference for Single and Interaction Treatment Effects Network~(CISI-Net) extended this approach to the multiple-treatment setting~\cite{murakami2025multipletreatmentscausaleffects}. 
To handle multiple treatments, this approach adopts a pairwise balancing strategy that aligns distributions between specific pairs of treatment patterns, while integrating a task embedding network to encode the unique contributions of single and interaction effects.

However, existing pairwise balancing methods face three critical limitations.
First, they struggle to achieve global alignment because reducing the discrepancy for one treatment-pattern pair may increase the discrepancies between other treatment-pattern pairs, leaving residual imbalance that can increase estimation variance.
Second, pairwise constraints generally do not encourage consistent local proximity structures across all treatment patterns, and such inconsistent geometric distortions can degrade heterogeneous effect estimation~\cite{cao2026pite,wang2025proximity,yao2018representation}.
Third, the quadratic computational complexity from matching all pairs makes these methods impractical for large-scale applications with numerous treatment combinations.

To address these challenges, we propose the Causal Inference for Heterogeneous Single and Interaction Treatment Effects Network~(CIHSI-Net), which minimizes global distributional discrepancies while preserving local proximity structures.
Central to CIHSI-Net is a novel Barycentric Fused Gromov-Wasserstein Balancing~(BFG-WB) objective, which aligns the representation distributions of each treatment pattern to a shared Wasserstein barycenter~\cite{agueh2011barycenters} using the Fused Gromov-Wasserstein~(FGW) discrepancy~\cite{titouan2019optimal,vayer2020fused}.
This approach achieves global alignment and preserves local proximity structures by accounting for both feature values and geometry, while reducing computational complexity from quadratic to linear.

Simulation results demonstrate that CIHSI-Net consistently outperforms state-of-the-art baselines in estimating both heterogeneous single and interaction effects.
Furthermore, ablation studies confirm the critical contribution of BFG-WB, and application to real-world marketing data validates the practical utility of the proposed framework.

\section{Related Work}
We review existing research from two perspectives: representation balancing methods and deep learning-based methods for multiple treatments.
First, deep learning-based representation balancing mitigates estimation variance arising from selection bias in observational data by reducing the discrepancy between the representations of the treatment and control groups.
Optimal Transport (OT)-based regularization, which globally aligns the representation distributions across treatment groups, has proven effective for estimating the heterogeneous causal effect~\cite{shalit2017estimating,li2021causal}.
Motivated by the reliance of counterfactual inference on local smoothness, recent extensions encourage the preservation of local proximity structures within each distribution by keeping input-space neighbors adjacent after balancing~\cite{yao2018representation,yao2019ace} or by guiding the cross-group alignment using local neighborhood geometries~\cite{wang2025proximity,cao2026pite}.

However, extending this pairwise balancing paradigm to the multiple-treatment setting faces fundamental limitations.
For global alignment, balancing each pair of treatment patterns separately~\cite{murakami2025multipletreatmentscausaleffects} allows the alignment of one pair to degrade that of another, leaving a residual imbalance~\cite{gong2022gromov,lian2020unsupervised}, with a cost quadratic in the number of patterns.
For local structure preservation, pairwise objectives promote geometric consistency only within pairs, not across all treatment patterns~\cite{liu2016structure,LU2025110864}.
Consequently, in a multiple-treatment setting, a unified framework achieving scalable global balancing with local geometry preservation remains an open challenge.

Second, instead of relying on pseudo-samples via data augmentation~\cite[e.g.][]{qian2021estimating}, we focus on architecture-driven approaches that intrinsically capture treatment interactions.
These methods typically rely on either separate outcome networks~\cite[e.g.][]{parbhoo2021ncore} or latent-variable generative models~\cite[e.g.][]{saini2019multiple, zou2020counterfactual}.
However, separate networks often yield unstable estimates for rare treatment patterns due to isolated parameters~\cite{chu2022hierarchical}, while latent-variable methods suffer reduced robustness when generative assumptions fail~\cite{rissanen2021critical}.
Consequently, there remains a need for an architecture that achieves efficient parameter sharing without strong generative assumptions.

To fulfill these needs, we propose a novel framework, CIHSI-Net.
The following section defines the problem setting and the causal estimands necessary for our framework.

\section{Preliminaries}
\subsection{Heterogeneous Effects under Multiple Treatments}
We formulate the problem of estimating heterogeneous single and interaction treatment effects under multiple treatments with the potential outcomes framework~\cite{rubin2005causal}. 
Let $\mathcal{D} = \{ (\boldsymbol{x}_i, \boldsymbol{t}_i, y_i) \}_{i=1}^{N}$ denote an observed dataset of $N$ independent units.
Here, $\boldsymbol{X} \in \mathcal{X} \subseteq \mathbb{R}^d$ denote the covariate vector and $ T \in \mathcal{T} = \{ 0, 1 \}^{K}$ denote the vector of $K$ simultaneous binary treatments, with realizations $\boldsymbol{x}$ and $\boldsymbol{t}$.
For each unit $i$, a potential outcome $Y_i(\boldsymbol{t}) \in \mathbb{R}$ exists for every $\boldsymbol{t} \in \mathcal{T}$, but only the factual outcome $y_i = Y_i(\boldsymbol{t}_i)$ is observed.
Let $Y(\boldsymbol{t})$ and $Y$ represent the corresponding population-level random variables for the potential and factual outcomes, respectively.

To identify causal effects of interest from the observed data, we adopt the following three assumptions commonly used in observational studies~\cite{imbens2015causal}. 
\begin{assumption}[Stable Unit Treatment Value Assumption]
\textit{
(1) no interference, meaning that the outcome of one unit is unaffected by the treatment assignments of other units; and (2) consistency of treatment, meaning that the potential outcomes correspond to well-defined and unique treatments~(i.e., $y_i = Y_i(\boldsymbol{t}_i)$).}
\end{assumption}

\begin{assumption}[Ignorability]
\textit{    
For any treatment pattern, the potential outcome is independent of the assigned treatment $\boldsymbol{T}$ given the observed covariates $\boldsymbol{X}$.
Formally, for all $\boldsymbol{t}$, 
$
Y(\boldsymbol{t}) \perp \boldsymbol{T} \mid \boldsymbol{X}.
$
}
\end{assumption}

\begin{assumption}[Overlap]
\textit{
Every unit has a non-zero probability of receiving any treatment pattern given its observed covariates.
Formally, for all $\boldsymbol{t}$ and $\boldsymbol{x}$, 
$
0 < P(\boldsymbol{T}=\boldsymbol{t} \mid \boldsymbol{X}=\boldsymbol{x}) < 1. 
$
}
\end{assumption}

Based on these assumptions, we define our estimands of interest: the Conditional Average Single Effect~(CASE) and the Conditional Average Interaction Effect~(CAIE).
Both estimands are defined using the conditional expected potential outcome given by 
\begin{equation}\label{eq:mu_function}
    \mu(\boldsymbol{x},\boldsymbol t) \coloneqq \mathbb{E} \bigl[Y(\boldsymbol t)\mid\boldsymbol X=\boldsymbol x\bigr].
\end{equation}
First, the CASE for treatment $k$ quantifies the marginal effect of applying treatment $k$ in isolation compared to the control group~(no treatment).
This estimand extends the standard conditional average treatment effect in the single-treatment setting~\cite{Abrevaya02102015} to the multi-treatment setting.
The CASE is defined as
\begin{equation}
\label{def:CASE}
\tau_{\mathrm{CASE}}(k, \boldsymbol{x}) \coloneqq
\mu(\boldsymbol{x}, \boldsymbol{t}_{+k}) - 
\mu(\boldsymbol{x}, \boldsymbol{0}), 
\end{equation}
where $\boldsymbol{t}_{+k}$ is the one-hot vector for treatment $k$.

Second, the CAIE quantifies interactions among a subset of treatments $S \subseteq \{ 1, \ldots, K \}$~(where $|S| \geq 2$).
We adopt the additive scale, which aligns with decision-making based on absolute outcome differences~(e.g., revenue), and extend the Average Interaction Effect~(AIE)~\cite{egami2019causal} to the conditional setting as follows.
\begin{equation}
\label{def:CAIE}
\tau_{\mathrm{CAIE}}(S,\boldsymbol{x}) \coloneqq
\sum_{Q \subseteq S} (-1)^{|S| - |Q|} \,
\mu\left(\boldsymbol{x}, \boldsymbol{t}_{(+Q)}\right), 
\end{equation}
where $\boldsymbol{t}_{(+Q)}$ denotes a vector with ones at indices in $Q$ and zeros elsewhere.
For example, in a two-treatment setting, $\tau_{\mathrm{CAIE}}(\{1,2\}, \boldsymbol{x})$ measures the deviation of the joint effect at $\boldsymbol{t} = (1,1)$ from the sum of the single effects of $(1, 0)$ and $(0, 1)$.

We now establish that the estimands defined above are identifiable from the joint distribution of the observed data $(\boldsymbol{x}, \boldsymbol{t}, y)$ under the stated assumptions.
\begin{prop}
\label{prop:identifiability}
Under Assumptions 1-3, CASE and CAIE are identifiable from the observed data.
\end{prop}
The proof is provided in Appendix~\ref{appendix:proof_identifiability}.

\subsection{Motivation for Representation Balancing}
\begin{figure}[t!]
  \centering
  \includegraphics[width=\linewidth, keepaspectratio]{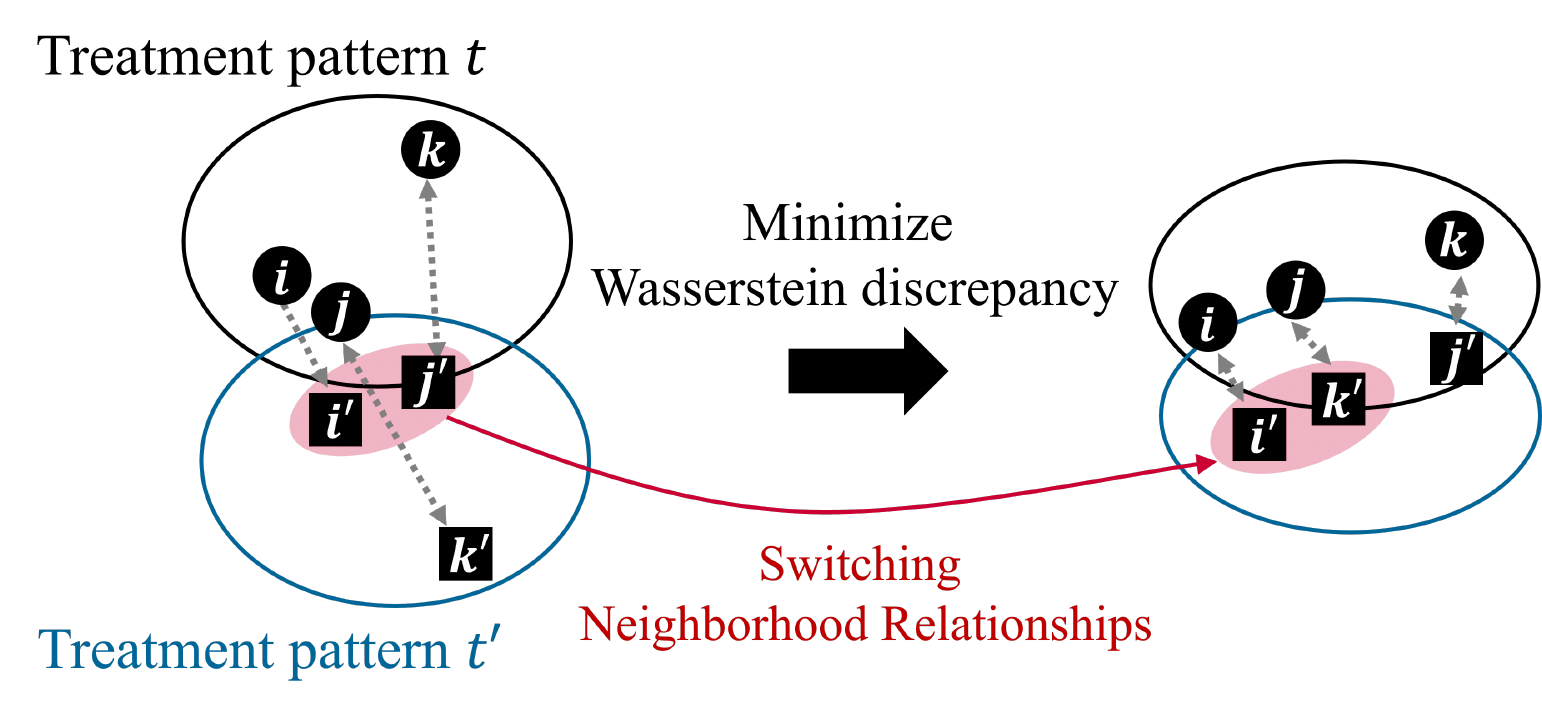}
  \caption{
    Illustration of neighborhood switching under Wasserstein-based representation balancing.
    The gray arrows show OT correspondences.
    Wasserstein minimization aligns the two distributions but can map neighboring points to distant locations, distorting within-distribution neighborhood relations; a structure-preserving match such as $j\!\leftrightarrow\!j'$ would be desirable.
  }
  \label{fig:viz_balancing_problem}
\end{figure}

Representation balancing mitigates estimation variance by minimizing distributional discrepancies across treatment groups.
A widely adopted metric is the Wasserstein distance~\cite[e.g.][]{cheng2022learning, shalit2017estimating}, which uses OT to measure the minimal transformation cost between distributions.

However, because the Wasserstein distance minimizes only the aggregate point-to-point transport cost, it does not explicitly account for local proximity structures within each distribution~\cite{wang2025proximity}.
Here, the local proximity structures refer to the relative geometric relationships among units within a specific treatment pattern.
Consequently, as illustrated in Figure~\ref{fig:viz_balancing_problem}, neighbors in one distribution may be mapped to scattered locations in another, potentially leading to structural collapse.

This disruption of local proximity structures degrades the accuracy of causal effect estimation.
Counterfactual prediction typically relies on the assumption of local smoothness, where units with similar representations yield similar outcomes~\cite{li2019locality, xu2012robustness}.
Therefore, accurate estimation of counterfactual outcomes $Y(\boldsymbol{t})$ benefits from (i) reducing distributional discrepancies~(global alignment) and (ii) ensuring that neighborhood consistency is maintained across treatment patterns~(local structure preservation).

This issue is particularly critical in the multiple-treatment setting.
When pairwise balancing is applied across numerous treatment patterns, local structure preservation becomes inconsistent, where the structural distortion may vary significantly depending on the pair being aligned.
Such non-uniformity destabilizes the reference neighbors used for inference, thereby degrading the estimation of heterogeneous single and interaction effects.

To address these challenges, we propose a novel framework that simultaneously achieves global alignment and local structure preservation.
We introduce a new architecture incorporating the FGW discrepancy to satisfy these dual objectives in the following section.

\section{CIHSI-Net: The Proposed Framework}
\subsection{Barycentric Fused Gromov-Wasserstein Balancing}
To overcome the limitations of pairwise balancing, we propose BFG-WB.
BFG-WB functions as a regularization term within our proposed CIHSI-Net, designed to achieve two critical objectives simultaneously: (i) reducing estimation variance induced by distributional imbalances globally by aligning all treatment representations to a common anchor, and (ii) preserving local proximity structures to ensure stable heterogeneous effect estimation.

Formally, BFG-WB minimizes the weighted sum of discrepancies between the representation distribution of each treatment pattern $\boldsymbol{t}$ and a shared Wasserstein barycenter $\boldsymbol{R}^{\ast}_{b}$.
Unlike standard approaches that use the Wasserstein distance, we employ the FGW discrepancy~\cite{titouan2019optimal,vayer2020fused}.
By integrating the feature-based Wasserstein distance with the structure-based Gromov-Wasserstein distance, the FGW discrepancy enables us to simultaneously achieve global distribution alignment~(Objective~(i)) and local structure preservation~(Objective~(ii)).
The BFG-WB regularization term $\mathcal{L}_{\phi}$ is defined as follows.
\begin{equation}
\label{eq:bfgwb_overview}
    \mathcal{L}_{\phi} = \sum_{\boldsymbol{t} \in \mathcal{T}} w_{\boldsymbol{t}} F\left(\boldsymbol{R}_{\boldsymbol{t}}, \boldsymbol{R}^{\ast}_{b}\right),
\end{equation}
where $\boldsymbol{R}_{b}^{\ast}$ is the Wasserstein barycenter, and
$\boldsymbol{R}_{\boldsymbol{t}}$ denotes the representation distribution on $\mathcal{R}$ for treatment pattern $\boldsymbol{t}$.
$F( \cdot, \cdot )$ measures the FGW discrepancy between two distributions, and $w_{\boldsymbol{t}} > 0$ is a weight satisfying $\sum_{\boldsymbol{t} \in \mathcal{T}} w_{\boldsymbol{t}} = 1$.
In the following, we detail the two core components: the Wasserstein barycenter $\boldsymbol{R}_{b}^{\ast}$ and the FGW discrepancy $F(\cdot, \cdot)$.

\paragraph{Wasserstein Barycenter.} 
Following the formulation in \cite{agueh2011barycenters}, we define the Wasserstein barycenter $\boldsymbol{R}_{b}^{\ast}$ as the centroid distribution that minimizes the weighted sum of squared Wasserstein distances from the representation distributions ${\boldsymbol{R}_{\boldsymbol{t}}}$ as follows.
\begin{equation}
\label{eq:def_barycenter}
\boldsymbol{R}^{\ast}_{b}
= \underset{\boldsymbol{R}_{b}} {\operatorname{argmin}}
\sum_{\boldsymbol{t} \in \mathcal{T}} \lambda_{\boldsymbol{t}} W^{2}_{2}(\boldsymbol{R}_{\boldsymbol{t}}, \boldsymbol{R}_{b}),
\end{equation}
where $W_2(\cdot,\cdot)$ is the 2-Wasserstein distance~(see Appendix~\ref{appendix:assumptions_and_notation} for its definition).
$\boldsymbol{R}_{b}$ denotes a barycenter candidate distribution, and $\lambda_{\boldsymbol{t}} \ge 0$ represents the weight for treatment pattern $\boldsymbol{t}$ satisfying $\sum_{\boldsymbol{t} \in \mathcal{T}} \lambda_{\boldsymbol{t}} = 1$.

Crucially, the introduction of the Wasserstein barycenter $\boldsymbol{R}_{b}^{\ast}$ resolves the optimization conflicts in pairwise balancing.
In standard pairwise approaches, minimizing the discrepancy for one pair of treatments may inadvertently increase it for another due to differing OT plans, which leads to persistent local imbalances associated with selection bias.
In contrast, our approach adopts a ``star-shaped''~(See Figure~\ref{fig:bfgwb_fig}) alignment strategy where all representation distributions are updated toward a single, fixed reference~(barycenter) $\boldsymbol{R}_{b}^{\ast}$.
This alignment strategy promotes coherent alignment and effectively mitigates estimation variance across all treatment patterns~(See Appendix~\ref{appendix:supporting_lemmas}).

\begin{figure}[t!]
  \centering
  \includegraphics[width=\linewidth, keepaspectratio]{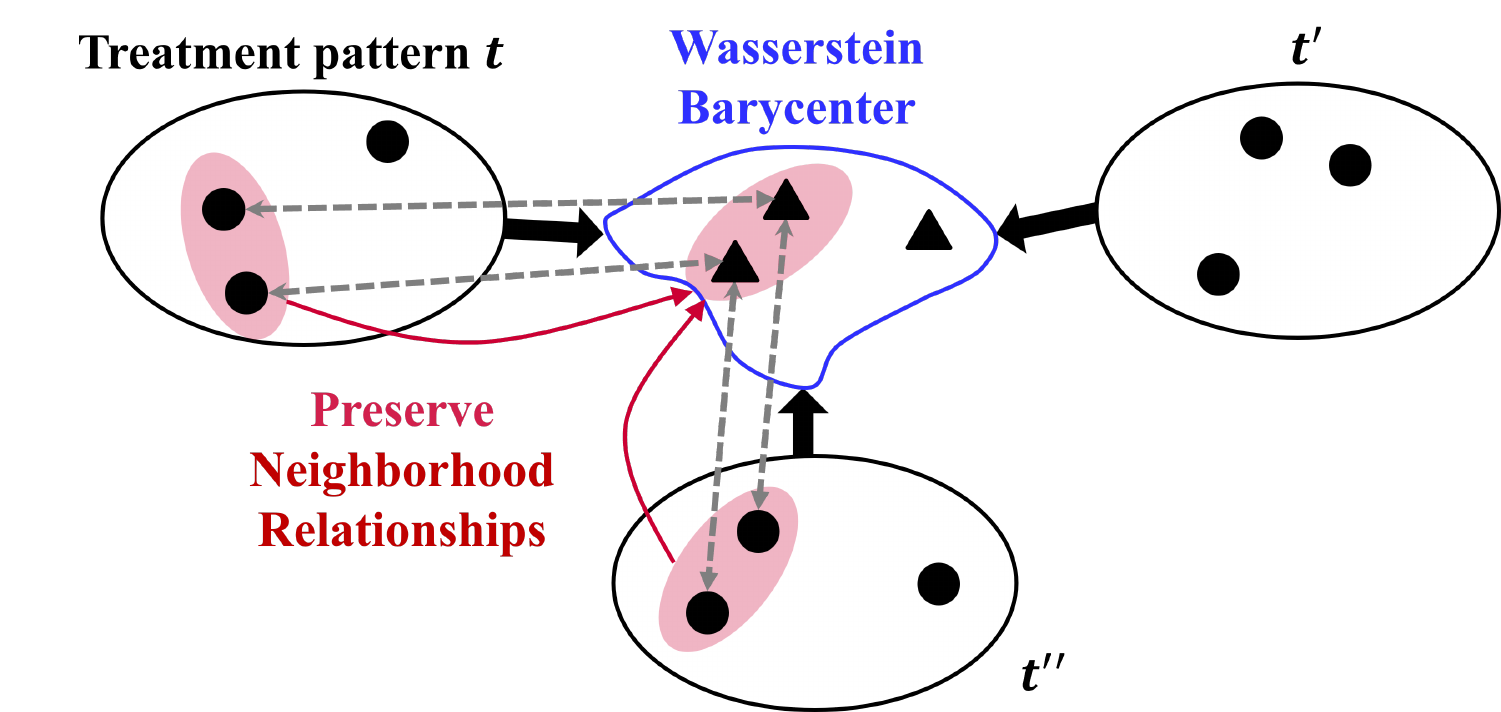}
  \caption{
    Conceptual illustration of BFG-WB. 
    The gray arrows show a subset of OT correspondences. 
    By evaluating the FGW discrepancy, the alignment is encouraged to preserve within-distribution neighborhood relations.
  }
  \label{fig:bfgwb_fig}
\end{figure}

\paragraph{Fused Gromov-Wasserstein Discrepancy.} 
With the barycenter $\boldsymbol{R}_{b}^{\ast}$ established as the global anchor, we employ the FGW discrepancy to measure the discrepancy between the representation distribution $\boldsymbol{R}_{\boldsymbol{t}}$ of each treatment pattern $\boldsymbol{t}$ and $\boldsymbol{R}_{b}^{\ast}$.
The FGW discrepancy integrates the Wasserstein distance~(for global distribution alignment), which evaluates feature values, and the Gromov-Wasserstein distance~(for local proximity structure preservation), which evaluates the geometric structure within distributions, under a single OT plan $\pi^{(\boldsymbol{t})}$.
Specifically, the FGW discrepancy $F(\boldsymbol{R}_{\boldsymbol{t}}, \boldsymbol{R}_{b}^{\ast})$ is defined as follows.

\begin{equation}
\label{eq:fgw_to_barycenter}
\begin{aligned}
F\!\left(\boldsymbol{R}_{\boldsymbol{t}}, \boldsymbol{R}^{\ast}_{b}\right)
=
\inf_{\pi^{(\boldsymbol{t})} \in \Pi(\boldsymbol{R}_{\boldsymbol{t}},\boldsymbol{R}^{\ast}_{b})}
\Bigg[
\eta
\int
c^{(\boldsymbol{t})}\!\left(\boldsymbol{r},\boldsymbol{z}\right)\, d\pi^{(\boldsymbol{t})}(\boldsymbol{r},\boldsymbol{z}) \quad
\\
+\;
(1-\eta)
\iint
d_b^{(\boldsymbol{t})}(\boldsymbol{r},\boldsymbol{r}',\boldsymbol{z},\boldsymbol{z}')
\, d\pi^{(\boldsymbol{t})}(\boldsymbol{r},\boldsymbol{z})\, d\pi^{(\boldsymbol{t})}(\boldsymbol{r}',\boldsymbol{z}')
\Bigg],
\end{aligned}
\end{equation}
where $\Pi(\boldsymbol{R}_{\boldsymbol{t}},\boldsymbol{R}^{\ast}_{b})$ denotes the set of all couplings on $\mathcal{R}\times\mathcal{R}$ with marginals $\boldsymbol{R}_{\boldsymbol{t}}$ and $\boldsymbol{R}^{\ast}_{b}$.
Here, $\boldsymbol{r},\boldsymbol{r}',\boldsymbol{z},\boldsymbol{z}'\in\mathcal{R}$ with $\boldsymbol{r},\boldsymbol{r}'\!\sim\!\boldsymbol{R}_{\boldsymbol{t}}$ and $\boldsymbol{z},\boldsymbol{z}'\!\sim\!\boldsymbol{R}^{\ast}_{b}$.
The first term accounts for global distribution alignment, where $c^{(\boldsymbol{t})}(\boldsymbol{r},\boldsymbol{z})=\|\boldsymbol{r}-\boldsymbol{z}\|_2$ measures the Euclidean distance between features.
The second term accounts for local proximity structure preservation.
Specifically, we define $d_b^{(\boldsymbol{t})}(\boldsymbol{r},\boldsymbol{r}',\boldsymbol{z},\boldsymbol{z}') := \bigl(d^{(\boldsymbol{t})}(\boldsymbol{r},\boldsymbol{r}') - d^{(b)}(\boldsymbol{z},\boldsymbol{z}')\bigr)^2$, where $d^{(\boldsymbol{t})}(\boldsymbol{r},\boldsymbol{r}')=\|\boldsymbol{r}-\boldsymbol{r}'\|_2$ and $d^{(b)}(\boldsymbol{z},\boldsymbol{z}')=\|\boldsymbol{z}-\boldsymbol{z}'\|_2$ represent the within-distribution distances~(local geometry) within the treatment representation distribution and the barycenter, respectively.
The hyperparameter $\eta\in(0,1]$ controls the trade-off between the first term~(global distribution alignment) and the second term~(local proximity structure preservation).

The second term in Equation~\eqref{eq:fgw_to_barycenter} plays an important role in our framework.
This term incurs a low cost only when pairs of points that are close in the representation space of treatment $\boldsymbol{t}$ are matched to pairs that are similarly close in the barycenter, thereby penalizing structure-blind matching based solely on feature similarity.
This mechanism mitigates distortions in neighborhood geometry by aligning each distribution to the local proximity structure of the barycenter.
As a result, it prevents counterfactual prediction from being performed using neighborhoods that mix units with different heterogeneity patterns.
This consistency significantly improves the estimation stability of heterogeneous causal effects, because such estimation depends on covariate-conditioned counterfactual prediction.

Beyond its methodological advantages in estimation variance reduction and structure preservation, BFG-WB offers a significant practical benefit: computational efficiency.
Whereas pairwise balancing requires comparing all pairs of treatment patterns, our star-shaped strategy compares each pattern only with the barycenter.
We formalize this advantage in the following proposition.
\begin{prop}[Computational Efficiency]
\label{prop:computational_efficiency}
Let $L = 2^K$ denote the number of treatment patterns derived from $K$ binary treatments, and let one evaluation denote the computation of a single OT-based discrepancy between two empirical distributions.
Per training step, standard pairwise balancing requires $\binom{L}{2} = O(L^2)$ evaluations, whereas BFG-WB requires at most $(L_b + 1)\,L = O(L)$ evaluations, where $L_b$ denotes a fixed upper bound on the number of barycenter-update iterations and is independent of $L$.
\end{prop}
A detailed derivation is provided in Appendix~\ref{appendix:computational_cost_remark}, and we empirically verify this scalability advantage through runtime comparisons in Appendix~\ref{appendix:exp_computing_speed}

\subsection{Theoretical Analysis}
In this section, we demonstrate that BFG-WB provides theoretically sound error control for estimating both CASE and CAIE from observational data.
Although accurate causal effect estimation requires controlling for prediction errors under all possible treatment patterns, counterfactual prediction errors cannot be minimized directly because only factual outcomes are observed.

To bridge this gap, we draw upon the theoretical framework of representation-based causal inference~\cite{wang2025proximity}.
A key insight from this literature is that the prediction error on the unobserved target domain~(counterfactuals) is theoretically bounded by the error on the observed source domain~(factuals) plus a discrepancy measure between their distributions.
Motivated by this principle, we derive upper bounds for the estimation errors of CASE and CAIE that are expressible solely through observational quantities.
Specifically, we demonstrate that these errors are bounded by the sum of the factual prediction error and the distributional discrepancies, which are explicitly minimized by the proposed BFG-WB.

The accuracy of the estimated CASE and CAIE is assessed using the integrated squared error with respect to the marginal covariate distribution $p(\boldsymbol{x})$.
These metrics quantify the expected squared deviation between the estimated and true causal effects, defined analogously to the Precision in Estimation of Heterogeneous Effects~(PEHE)~\cite{hill2011bayesian} used in single-treatment studies.
Formally, the estimation errors of CASE and CAIE are defined as follows.
\begin{equation}
\label{eq:def_case_error}
    \epsilon_{\mathrm{CASE}}(k) = \int_{\mathcal{X}}\left( \hat\tau_{\mathrm{CASE}}\left(k,\boldsymbol{x} \right) - \tau_{\mathrm{CASE}}\left(k,\boldsymbol{x} \right) \right)^2 p(\boldsymbol{x}) d\boldsymbol{x},
\end{equation}
\begin{equation}
\label{eq:def_caie_error}
    \epsilon_{\mathrm{CAIE}}({S}) = \int_{\mathcal{X}}\left( \hat\tau_{\mathrm{CAIE}}\left(S,\boldsymbol{x} \right) - \tau_{\mathrm{CAIE}}\left(S,\boldsymbol{x} \right) \right)^2 p(\boldsymbol{x}) d\boldsymbol{x},
\end{equation}
where $\hat{\tau}_{\mathrm{CASE}}(\cdot)$ and $\hat{\tau}_{\mathrm{CAIE}}(\cdot) $ denote the estimated CASE and CAIE.
For each treatment pattern $\boldsymbol{t}$, the expected prediction error of the factual outcome is defined as follows.
\begin{equation}
\epsilon_{\mathrm{F}}^{(\boldsymbol{t})}
=
\int_{\mathcal{X}}
l\left( \boldsymbol{x}, \boldsymbol{t} \right)
\, p(\boldsymbol{x} \mid \boldsymbol{t}) \,p(\boldsymbol{t})\, d\boldsymbol{x},
\end{equation}
where 
$
  l(\boldsymbol{x},\boldsymbol{t})
  =
  \int \bigl(Y(\boldsymbol{t})-\hat\mu(\boldsymbol{x}, \boldsymbol{t})\bigr)^2\,p\bigl(Y(\boldsymbol{t})\mid \boldsymbol{x}\bigr)\,dY(\boldsymbol{t})
$
denotes the conditional expected squared prediction loss, and $\hat{\mu}(\boldsymbol{x},\boldsymbol{t})$ denotes the estimate of $\mu(\boldsymbol{x},\boldsymbol{t})$.

Formally, we establish the upper bounds for $\epsilon_{\text{CASE}}$ and $\epsilon_{\text{CAIE}}$ as follows.
\begin{thm}[Upper Bound for CASE]
\label{thm:upper_bound_case}
Suppose that Assumptions 1–3 and the auxiliary conditions in Appendix~\ref{appendix:assumptions_and_notation} hold~(in particular, the expected squared loss $l(\boldsymbol{x},\boldsymbol{t})$ is $B_{\phi}$-Lipschitz in the representation space), and that the weights in Equation~\eqref{eq:bfgwb_overview} are uniform, i.e., $w_{\boldsymbol{t}}=2^{-K}$. 
Then, for any $k \in \{1,\ldots,K\}$,
\begin{equation*}
% \label{eq:case}
\epsilon_{\mathrm{CASE}}(k) \le 2 \Biggl\{ \frac{1}{p(\boldsymbol{t}_{+k})}\epsilon_{\mathrm{F}}^{(\boldsymbol{t}_{+k})} + \frac{1}{p(\boldsymbol{0})}\epsilon_{\mathrm{F}}^{(\boldsymbol{0})}
+ \frac{2^{2K}}{\eta} B_{\phi}\mathcal{L}_{\phi} \Biggr\}. 
\end{equation*}
\end{thm}

\begin{thm}[Upper Bound for CAIE]
\label{thm:upper_bound_caie}
Under the same assumptions as Theorem~\ref{thm:upper_bound_case}, consider an interaction set $S \subseteq \{1,\ldots,K\}$ with $|S|\ge 2$.
Let $(a_{\boldsymbol{t}})_{\boldsymbol{t}\in\mathcal{T}}$  be a constant vector reflecting the combinatorial structure of $S$~($a_{\boldsymbol{t}}\in\{-1,0,1\}$). 
The estimation error $\epsilon_{\mathrm{CAIE}}(S)$ is bounded by:
\begin{equation*}
% \label{eq:caie}
\begin{aligned}
\epsilon_{\mathrm{CAIE}}(S)
&\le
\left(\sum_{\boldsymbol{t} \in \mathcal{T}} a_{\boldsymbol{t}}^2 \right)
\Biggl\{
\sum_{\boldsymbol{t} \in \mathcal{T}} \frac{1}{p(\boldsymbol{t})}\,
\epsilon_{\mathrm{F}}^{(\boldsymbol{t})}
\\
&\qquad\qquad
+ \frac{2^{K+1}}{\eta}\, B_{\phi}\,\left(2^K - 1\right)\, \mathcal{L}_{\phi}
\Biggr\}.
\end{aligned}
\end{equation*}
\end{thm}

The detailed proofs for Theorem~\ref{thm:upper_bound_case} and Theorem~\ref{thm:upper_bound_caie} are provided in Appendix~\ref{appendix:proof_theorem_upper}.
These theorems demonstrate that the estimation errors for CASE and CAIE are theoretically upper-bounded by the sum of the factual prediction error $\epsilon_{\mathrm{F}}^{(\boldsymbol{t})}$ and the distributional discrepancy term $\mathcal{L}_{\phi}$.
Crucially, these bounds explicitly highlight the inverse probability term $1 / p(\boldsymbol{t})$, which formalizes the vulnerability of causal inference to rare treatment patterns.
Furthermore, the dependence on $2^K$ in these bounds reflects the inherent combinatorial complexity of estimating causal effects over all treatment patterns.
By aligning all treatment-pattern representation distributions to a single fixed anchor $\boldsymbol{R}_{b}^{\ast}$, our star-shaped strategy efficiently minimizes all pairwise Wasserstein discrepancies via the triangle inequality.
Simultaneously, the Gromov-Wasserstein component of the FGW discrepancy preserves local proximity structures within the representation space.
This structural preservation makes the Lipschitz continuity assumption~($B_{\phi}$) more plausible in practice, supporting counterfactual inference from nearby factual samples. 
Together, these results show that BFG-WB provides both computational efficiency and a principled mechanism for controlling the estimation errors of CASE and CAIE.

\subsection{Network Architecture and Training Objective}
\label{subsec:network}
We implement our proposed method using CIHSI-Net, a neural network architecture designed to achieve both accurate outcome prediction and robust representation balancing.
As illustrated in Figure~\ref{fig:model_architecture}, the network consists of three components: a representation learning network $\phi: \mathbb{R}^d \rightarrow \mathbb{R}^{d_r}$ regularized by BFG-WB, a task embedding network $t_{w}:\mathbb{R}^{K} \rightarrow \mathbb{R}^{d_t}$ that captures treatment similarities, and the outcome prediction network $h: \mathbb{R}^{d_r+d_t} \rightarrow \mathbb{R}$ that estimates outcomes from the concatenated features.
\begin{figure}[t]
  \centering
  \includegraphics[width=\linewidth, keepaspectratio]{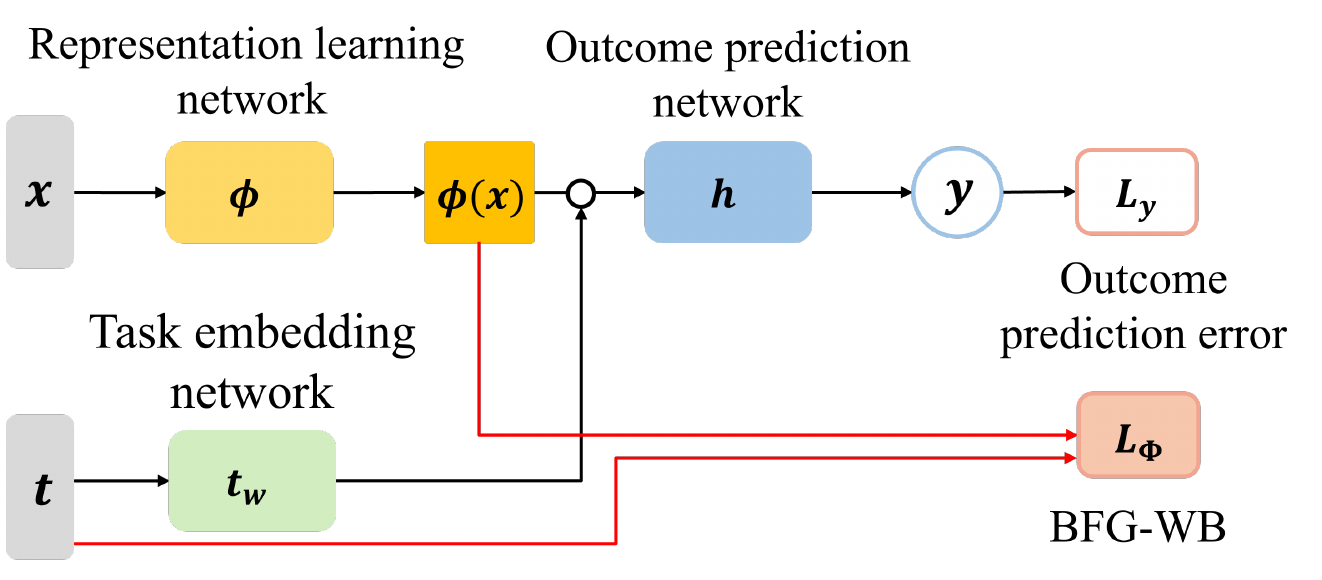}
  \caption{
  The architecture of CIHSI-Net.
  The model consists of three components: the representation learning network $\phi$, the task embedding network $t_w$, and the outcome prediction network $h$. 
  BFG-WB is a regularization term applied to the latent representation $\phi(\boldsymbol{x})$. 
  }
  \label{fig:model_architecture}
\end{figure}

The entire network is optimized by minimizing a composite loss function $\mathcal{L}$ that balances factual prediction accuracy with the distributional alignment constraints derived in our theoretical analysis as follows.
\begin{equation}
\label{eq:total_loss} 
\mathcal{L} = \mathcal{L}_y + \alpha \mathcal{L}_{\phi} + \beta \| w \|_2^2, 
\end{equation}
where the first term $\mathcal{L}_y$ represents the factual outcome prediction error, the second term $\mathcal{L}_\phi$ is the BFG-WB regularization term, and the third term is an L2 regularization term applied to the network weights.
The coefficients $\alpha$ and $\beta$ are hyperparameters that control the strength of the corresponding terms.

The outcome prediction loss $\mathcal{L}_y$ is designed to align with the error bounds established in Theorems~\ref{thm:upper_bound_case} and \ref{thm:upper_bound_caie}.
Guided by our theoretical analysis, which reveals that prediction errors on rare treatment patterns are amplified by the inverse probability $1 / p(\boldsymbol{t})$, we employ an inverse-frequency weighted mean squared error to counteract this imbalance as follows.
\begin{equation} 
\label{eq:weighted_loss} 
\mathcal{L}_y = \frac{1}{N} \sum_{i=1}^{N} \frac{1}{\hat{p}(\boldsymbol{t}_i)} ( y_i - \hat{y}_i )^2, 
\end{equation}
where $\hat{p}({\boldsymbol{t}_i}) = \frac{1}{N} \sum_{j=1}^{N} \mathbb{I}[\boldsymbol{t}_i = \boldsymbol{t}_j]$ represents the empirical frequency of the treatment pattern assigned to unit $i$.

\section{Simulation Experiments}
\label{main:simulation}
We evaluate our CIHSI-Net using simulation datasets to (i) compare estimation accuracy with state-of-the-art baselines and (ii) clarify the contributions of its components through ablation studies.

\paragraph{Datasets.} 
To assess estimation accuracy under selection bias and heterogeneous effects, we generate two types of simulation datasets following the protocol in \citet{murakami2025multipletreatmentscausaleffects}.
Simulation 1~(With Interactions) is generated with complex interaction effects among treatments and serves as the primary benchmark for evaluating CASE and CAIE estimation. 
Simulation 2~(Without Interactions) is generated without interaction effects and tests the model's robustness by verifying that it does not produce spurious interaction effects when none exist. 
We fix the number of treatments at $K=3$ and the sample size at $N=50,000$.
Detailed data generation processes are provided in Appendix~\ref{appendix:dgp}.

\paragraph{Experimental Setup.} 
We compare CIHSI-Net with three representative multi-treatment causal inference methods: Task Embedding–based Causal Effect Variational Autoencoder~(TECE-VAE)~\cite{saini2019multiple}, Neural Counterfactual Relation Estimation~(NCoRE)~\cite{parbhoo2021ncore}, and CISI-Net~\cite{murakami2025multipletreatmentscausaleffects}.
Implementation details for all baselines and the hyperparameters of CIHSI-Net are provided in Appendix~\ref{appendix:implement_details}. 
The same network hyperparameter settings are applied across all experiments.
For CIHSI-Net, we set $\alpha$ in Equation~\eqref{eq:total_loss} to 1.0 and $\eta$ in Equation~\eqref{eq:fgw_to_barycenter} to 0.6 for comparative experiments, while fixing $\alpha = 1.0$ and varying $\eta$ for ablation studies.

All models, including baselines, are trained using standard optimization protocols with 60\% training, 10\% validation, and 30\% test data.
A sensitivity analysis regarding the regularization coefficients is presented in Appendix~\ref{appendix:sensitivity_analysis}.

\paragraph{Evaluation Metrics.}
Performance is evaluated using the integrated squared errors for CASE and CAIE, as defined in Equations~\eqref{eq:def_case_error} and \eqref{eq:def_caie_error}.
To ensure statistical reliability, we report average metrics over 100 independent runs.

\paragraph{Accuracy Comparison against Baselines.}
Table~\ref{tab:simulation_data_result_all} shows ${\epsilon_{\mathrm{CASE}}}$ and ${\epsilon_{\mathrm{CAIE}}}$ of the proposed and baseline methods across two simulation datasets.
Across all simulation settings, the proposed method consistently achieves the lowest estimation errors for both CASE and CAIE, outperforming all baseline methods.
CIHSI-Net maintains high accuracy regardless of whether interaction effects are present, and its advantage over the baselines is most pronounced for ${\epsilon_{\mathrm{CAIE}}}$, indicating that BFG-WB is particularly beneficial for interaction treatment effect estimation.
Further analysis in Appendix~\ref{appendix:learned_rep} demonstrates that our CIHSI-Net effectively reduces distributional discrepancies~(measured by Wasserstein distance) while preserving local proximity structures~(measured by Gromov-Wasserstein distance), thereby contributing to superior estimation performance.
Moreover, CIHSI-Net maintains its CASE and CAIE estimation advantage as $K$ increases to eight, demonstrating its scalability~(see Appendix~\ref{appendix:simulation_scalability}), and achieves the best overall estimation performance in the semi-synthetic setting (see Appendix~\ref{appendix:evaluation_semi_synthetic}).

\begin{table*}[tb]
  \centering
  \small
  \begin{tabular}{llccc|cccc}
  \hline
  \multicolumn{2}{l}{} & \multicolumn{3}{c|}{${\epsilon_{\mathrm{CASE}}}$} & \multicolumn{4}{c}{${\epsilon_{\mathrm{CAIE}}}$} \\
  \cline{3-5}\cline{6-9}
  Sim.&Method& $k=1$ & $k=2$ & $k=3$ & $S=\{1,2\}$ & $S=\{2,3\}$ & $S=\{1,3\}$ & $S=\{1,2,3\}$ \\
  \hline
  \multirow{5}{*}{1}
   & TECE-VAE & 1.82 $\pm$ 0.09 & 1.91 $\pm$ 0.10 & 1.80 $\pm$ 0.14 & 3.58 $\pm$ 0.11 & 3.27 $\pm$ 0.10 & 3.30 $\pm$ 0.10 & 6.80 $\pm$ 0.18 \\
   & NCoRE    & 0.28 $\pm$ 0.06 & 0.30 $\pm$ 0.06 & 0.25 $\pm$ 0.05 & 0.52 $\pm$ 0.10 & 0.33 $\pm$ 0.10 & 0.37 $\pm$ 0.10 & 0.56 $\pm$ 0.14 \\
   & CISI-Net & 0.22 $\pm$ 0.04 & 0.24 $\pm$ 0.03 & 0.19 $\pm$ 0.06 & 0.31 $\pm$ 0.07 & 0.11 $\pm$ 0.04 & 0.15 $\pm$ 0.08 & 0.31 $\pm$ 0.08 \\
   & \textbf{CIHSI-Net} & \textbf{0.19 $\pm$ 0.04} & \textbf{0.21 $\pm$ 0.04} & \textbf{0.18 $\pm$ 0.10} & \textbf{0.28 $\pm$ 0.05} & \textbf{0.08 $\pm$ 0.05} & \textbf{0.12 $\pm$ 0.04} & \textbf{0.24 $\pm$ 0.04} \\
  \midrule
  \multirow{5}{*}{2}
   & TECE-VAE & 1.60 $\pm$ 0.05 & 1.65 $\pm$ 0.09 & 1.58 $\pm$ 0.16 & 3.01 $\pm$ 0.05 & 2.65 $\pm$ 0.08 & 2.69 $\pm$ 0.09 & 5.50 $\pm$ 0.15 \\
   & NCoRE    & 0.26 $\pm$ 0.05 & 0.28 $\pm$ 0.09 & 0.24 $\pm$ 0.07 & 0.34 $\pm$ 0.12 & 0.30 $\pm$ 0.09 & 0.27 $\pm$ 0.10 & 0.35 $\pm$ 0.12 \\
   & CISI-Net & 0.16 $\pm$ 0.04 & 0.19 $\pm$ 0.07 & \textbf{0.16 $\pm$ 0.05} & 0.08 $\pm$ 0.08 & 0.07 $\pm$ 0.04 & 0.07 $\pm$ 0.05 & 0.11 $\pm$ 0.09 \\
   & \textbf{CIHSI-Net} & \textbf{0.14 $\pm$ 0.04} & \textbf{0.17 $\pm$ 0.03} & \textbf{0.16 $\pm$ 0.12} & \textbf{0.06 $\pm$ 0.04} & \textbf{0.06 $\pm$ 0.03} & \textbf{0.06 $\pm$ 0.03} & \textbf{0.09 $\pm$ 0.04} \\
  \bottomrule
  \end{tabular}
  \caption{
  Comparison of mean estimation errors and standard deviations for ${\epsilon_{\mathrm{CASE}}}$ and ${\epsilon_{\mathrm{CAIE}}}$ in two simulation datasets. 
  Here, $k \in \{1,2,3\}$ indexes individual treatments, and $S \subseteq \{1,2,3\}$ (with $|S| \geq 2$) denotes treatment combinations.
  }
  \label{tab:simulation_data_result_all}
\end{table*}

\paragraph{Ablation Study.}
To verify the contribution of each component in our proposed method, we conducted an ablation study by varying the alignment strategy~(with or without Wasserstein barycenter) and the discrepancy metric~(Wasserstein distance~(W), Gromov-Wasserstein discrepancy~(GW), and FGW discrepancy~(FGW)).
Table~\ref{tab:simulation_ablation_average} summarizes the results.
Detailed results for $\epsilon_{\mathrm{CASE}}$ and $\epsilon_{\mathrm{CAIE}}$ are reported in Appendix~\ref{appendix:ablation_results_details}.

First, the combination of the Wasserstein barycenter and FGW discrepancy~(No.7) consistently achieves the lowest average estimation errors for both CASE and CAIE.
Second, methods using the Wasserstein barycenter generally outperform their pairwise counterparts~(compare No.5-7 with No.2-4).
This finding suggests that pairwise discrepancy minimization is insufficient for globally reducing representation imbalance across treatment patterns, whereas barycentric alignment more effectively mitigates imbalance-driven estimation variance.
Third, regarding the discrepancy metric, while the FGW discrepancy fails to improve or even yields slightly inferior results compared to the Wasserstein distance in the pairwise setting~(No.4 vs No.2), it improves performance when combined with the Wasserstein barycenter~(No.7 vs No.5).
These findings suggest that the structural preservation capability of FGW discrepancy is effective only when supported by the global alignment provided by the Wasserstein barycenter.

\begin{table}[tb]
    \centering
    \small
    \begin{tabular}{lcccc}
    \toprule
    No. & WB & Penalty type 
    & Average $\epsilon_{\mathrm{CASE}}$ 
    & Average $\epsilon_{\mathrm{CAIE}}$ \\
    \midrule
    1 & \xmarkgray & None & 0.27 $\pm$ 0.06 & 0.24 $\pm$ 0.02\\
    2 & \xmarkgray & W    & 0.22 $\pm$ 0.03 & 0.22 $\pm$ 0.03\\
    3 & \xmarkgray & GW   & 0.25 $\pm$ 0.06 & 0.24 $\pm$ 0.03 \\
    4 & \xmarkgray & FGW  & 0.22 $\pm$ 0.02 & 0.23 $\pm$ 0.02 \\
    \midrule
    5 & \cmark     & W    & 0.22 $\pm$ 0.04 & 0.21 $\pm$ 0.02 \\
    6 & \cmark     & GW   & 0.22 $\pm$ 0.04 & 0.22 $\pm$ 0.03\\
    7 & \cmark     & FGW  & \textbf{0.19 $\pm$ 0.04} & \textbf{0.18 $\pm$ 0.02} \\
    \bottomrule
    \end{tabular}
    \caption{
    Ablation study results on simulation dataset 1.
    Average $\epsilon_{\mathrm{CASE}}$ and Average $\epsilon_{\mathrm{CAIE}}$ denote the means of the reported $\epsilon_{\mathrm{CASE}}$ and $\epsilon_{\mathrm{CAIE}}$ values, respectively.
    WB indicates whether the Wasserstein barycenter is used.
    Penalties: Wasserstein~(W), Gromov-Wasserstein~(GW), Fused GW~(FGW), and None.
    }
    \label{tab:simulation_ablation_average}
\end{table}

\section{Application to Multiple Marketing Promotions}
\label{main:application_marketing_promotions}

We apply CIHSI-Net to a real-world dataset from a mobile payment platform to validate its effectiveness in capturing heterogeneous causal effects.
The dataset includes three simultaneously conducted promotions: two offline promotions by the same merchant group~(CP\textsubscript{1}, CP\textsubscript{2}) and one online promotion by a different merchant group~(CP\textsubscript{3}).
The outcome is defined as the total payment amount during the promotion period, standardized for confidentiality. 
User covariates include 71 features derived from service usage history and demographics.
To analyze heterogeneity, we stratify users into 11 groups based on their payment amount in the pre-promotion month. 
Hyperparameters are set to the best-performing configuration from the simulation study~(Section~\ref{main:simulation}). 
Detailed dataset statistics and preprocessing procedures are provided in Appendix~\ref{appendix:sub_real_world_data_details}.

Figure~\ref{fig:empirical_result} shows the estimated CASE and CAIE across user groups obtained from a single algorithm run.
\begin{figure}[t]
  \centering
  \includegraphics[width=\linewidth, keepaspectratio]{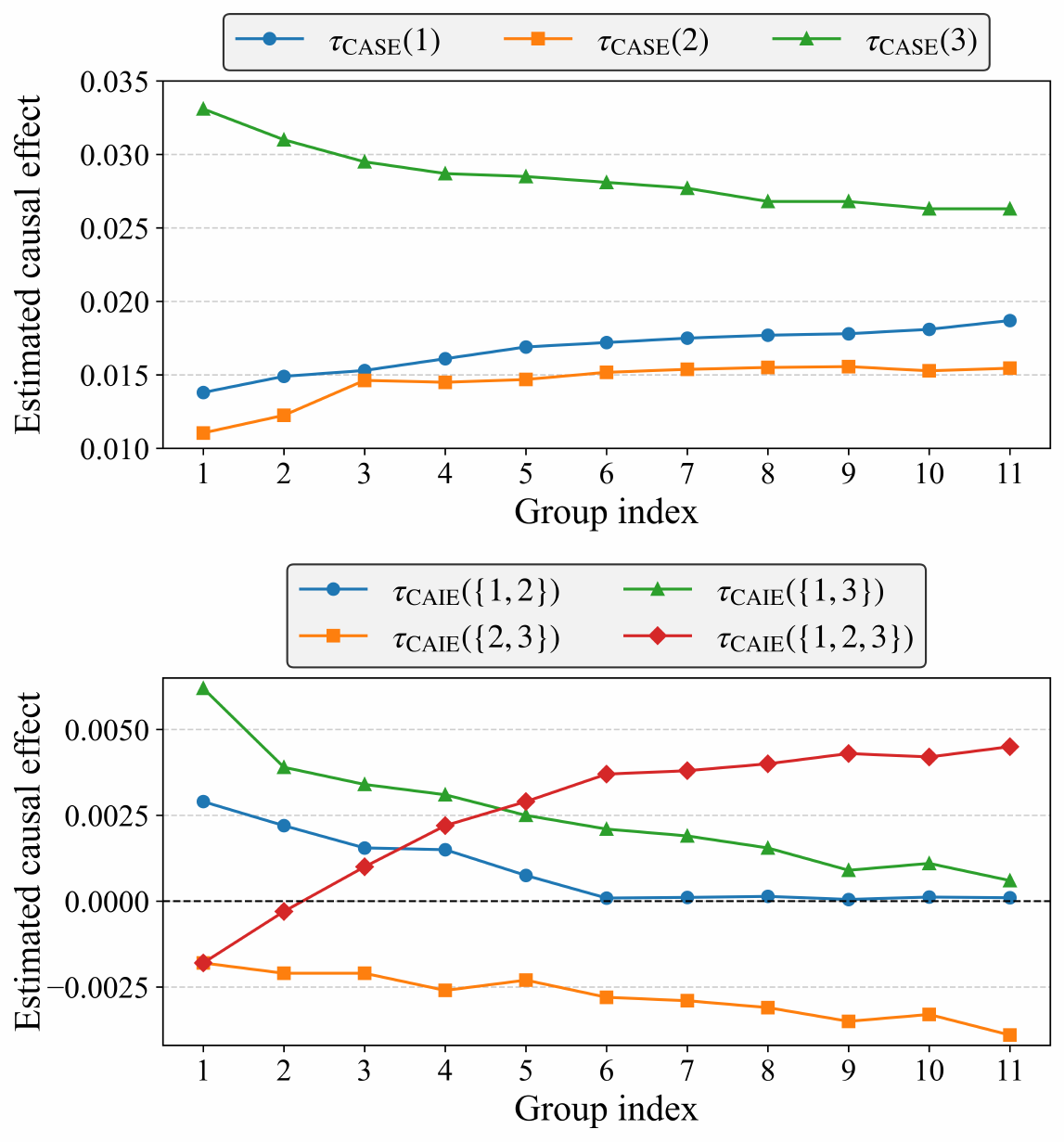}   
  \caption{
    Estimated CASE~(top) and CAIE~(bottom) across user groups stratified by pre-promotion service usage.
  }
  \label{fig:empirical_result}
\end{figure}
The top panel in Figure~\ref{fig:empirical_result} shows that all estimated CASEs are positive.
The online promotion CP$\textsubscript{3}$ exhibits the strongest effect for low-usage users.
This result is consistent with findings that online incentives effectively encourage increased service spending, even among users with little service experience~\cite{blake2015consumer, langen2023causal}.
In contrast, the offline promotions~(CP\textsubscript{1}, CP\textsubscript{2}) show modest increases for higher-usage groups.
These results suggest that user groups for whom effects are strongly observed may differ across channels and promotion designs.

The bottom panel in Figure~\ref{fig:empirical_result} highlights the capability of CIHSI-Net to capture complex interaction structures.
Whereas the estimated interaction between same-group promotions ($\tau_{\text{CAIE}}(\{1,2\})$) remains positive but declines with usage, cross-group interactions show mixed patterns.
Notably, the estimated three-way interaction $\tau_{\text{CAIE}}(\{1,2,3\})$ shifts from negative for low-usage users to positive for high-usage users.
This result suggests that, for users who regularly use the service, conducting more promotions simultaneously may yield complementary effects, whereas for low-usage users, presenting multiple incentives may divide attention~(like choice overload~\cite{CHERNEV2015333}) and result in a negative interaction effect.
Additional sensitivity and uncertainty analyses in Appendix~\ref{appendix:sensitivity_and_uncertainty_analysis_on_datasetA} show that the qualitative CASE and CAIE patterns remain stable across the BFG-WB hyperparameters $\alpha$ and $\eta$ and are broadly consistent under bootstrap confidence intervals.
These results highlight the importance of capturing treatment effect heterogeneity and demonstrate the practical utility of CIHSI-Net in revealing complex behavioral mechanisms, such as coexisting synergistic and cannibalistic interactions that vary substantially across user groups.

To further validate the generalizability of our CIHSI-Net, we extended our analysis to a second real-world scenario involving promotions from competing merchants.
Due to space constraints, detailed results demonstrating the detection of cannibalization are provided in Appendix~\ref{appendix:additional_real_world_results}.

\section{Conclusion}
This study addresses the limitations of pairwise balancing in multi-treatment causal inference by proposing Causal Inference for Heterogeneous Single and Interaction Treatment Effects Network~(CIHSI-Net).
By introducing Barycentric Fused Gromov-Wasserstein Balancing~(BFG-WB), our framework achieves scalable global alignment via a Wasserstein barycenter while preserving local proximity structures through the FGW discrepancy.
Experiments on simulation datasets demonstrate that CIHSI-Net outperforms baselines in estimating heterogeneous single and interaction effects, and a real-world marketing application illustrates its practical utility in multi-treatment settings.
Several directions remain for future work. 
First, a promising extension is to replace the Wasserstein barycenter with an FGW barycenter to unify global alignment and structural preservation, enabling more consistent alignment of feature distributions and within-distribution geometry.
Second, it will be valuable to connect CIHSI-Net to downstream decision-making, such as uplift-based allocation for combinatorial treatment.

\bibliography{aaai2027}

\clearpage
\onecolumn

\appendix
\begin{center}
    {\Large\bfseries Barycentric Fused Gromov-Wasserstein Balancing \\ for Causal Inference under Multiple Treatments~(Supplementary Material)\par}
    \vspace{3em}
\end{center}

\renewcommand{\qedsymbol}{$\blacksquare$}
\renewcommand{\thetable}{A\arabic{table}}
\setcounter{table}{0}
\newcommand{\bfcell}[1]{\textbf{\boldmath #1}}
\renewcommand{\thefigure}{A\arabic{figure}}
\makeatletter
\def\section{\@startsection{section}{1}{\z@}{-10pt plus -2pt minus -2pt}{4pt plus 2pt minus 2pt}{\Large\bfseries\centering}}
\def\subsection{\@startsection{subsection}{2}{\z@}{-8pt plus -2pt minus -2pt}{4pt plus 2pt minus 2pt}{\large\bfseries}}
\def\subsubsection{\@startsection{subsubsection}{3}{\z@}{-6pt plus -2pt minus -2pt}{4pt plus 2pt minus 2pt}{\normalsize\bfseries}}
\makeatother
\setcounter{secnumdepth}{3}
\renewcommand{\thesection}{\Alph{section}}
\renewcommand{\thesubsection}{\thesection.\arabic{subsection}}
\renewcommand{\thesubsubsection}{\thesubsection.\arabic{subsubsection}}

% Equation numbers: A.1, A.2, ..., B.1, B.2, ...
\numberwithin{equation}{section}

\setcounter{assumption}{0}
\setcounter{thm}{0}
\setcounter{figure}{0}
\setcounter{table}{0}

\section{Proofs of Theoretical Results}

\subsection{Proof of Proposition 1}
\label{appendix:proof_identifiability}
\begin{lemma}
\label{lemma:identifiability_of_cepo}
Under Assumptions 1-3, the conditional expected potential outcome $\mu(\boldsymbol{x}, \boldsymbol{t})$ is identifiable from the observed data.
\end{lemma}

\begin{proof}
The identifiability of the conditional expected potential outcome is a standard result in the potential outcome framework.
The equality holds directly from Assumptions~1 and 2, while Assumption~3 ensures the conditional expected potential outcome is well-defined.
For a detailed derivation, see, for example, \citet{murakami2025multipletreatmentscausaleffects}.

\end{proof}

Using Lemma~\ref{lemma:identifiability_of_cepo}, we prove that CASE and CAIE are identifiable.
\begin{proof}
Recall that the CASE for treatment $k$, denoted as $\tau_{\text{CASE}} (k, \boldsymbol{x})$, is defined as 
\begin{equation*}
\tau_{\mathrm{CASE}}(k, \boldsymbol{x}) \coloneqq
\mu(\boldsymbol{x}, \boldsymbol{t}_{+k}) - 
\mu(\boldsymbol{x}, \boldsymbol{0}), 
\end{equation*}
where $\boldsymbol{t}_{+k}$ is the one-hot vector for treatment $k$.
Similarly, recall that the CAIE for a set of treatments $S$, denoted as $\tau_{\mathrm{CAIE}}(S,\boldsymbol{x})$, is defined as
\begin{equation*}
\tau_{\mathrm{CAIE}}(S,\boldsymbol{x}) \coloneqq
\sum_{Q \subseteq S} (-1)^{|S| - |Q|} \,
\mu\left(\boldsymbol{x}, \boldsymbol{t}_{(+Q)}\right) .
\end{equation*}
where $\boldsymbol{t}_{(+Q)}$ corresponds to the treatment vector where treatments in the subset $Q$ are active.

According to Lemma~\ref{lemma:identifiability_of_cepo}, for any treatment vector $\boldsymbol{t} \in \{0, 1\}^{K}$, the term $\mu(\boldsymbol{x}, \boldsymbol{t})$ is identifiable from the observed data as $\mathbb{E} [Y \mid \boldsymbol{X} = \boldsymbol{x}, \boldsymbol{T} = \boldsymbol{t}]$.
Both $\tau_{\mathrm{CASE}}(k, \boldsymbol{x})$ and $\tau_{\mathrm{CAIE}}(S,\boldsymbol{x})$ are constructed as linear combinations of $\mu(\boldsymbol{x}, \boldsymbol{t})$ for specific values of $\boldsymbol{t}$.
Since every term $\mu(\boldsymbol{x}, \boldsymbol{t})$ in these equations is identifiable, any linear combination of them is also identifiable from the observed data.
Therefore, CASE and CAIE are identifiable.

\end{proof}

\subsection{Proof of Proposition 2}
\label{appendix:computational_cost_remark}
\begin{proof}
Let $L = 2^K$ denote the total number of treatment patterns
derived from $K$ binary treatments.

Standard pairwise balancing evaluates an OT-based discrepancy
for every unordered pair of distinct treatment-pattern
representation distributions.
Therefore, the total number of evaluations per training step is
\[
N_{\mathrm{pair}}
=
\binom{L}{2}
=
\frac{L(L-1)}{2}
=
O(L^2).
\]

For BFG-WB, the total number of evaluations consists of two
components.
First, for Wasserstein barycenter estimation, each barycenter-update iteration solves one OT problem between the current barycenter and each of the $L$ input treatment-pattern distributions~\citep{cuturi2014fast}.
Therefore, under a fixed iteration budget of at most $L_b$ barycenter-update iterations, the number of OT evaluations required for barycenter estimation satisfies $N_{\mathrm{bary}} \leq L_b L$ where $L_b$ is independent of $L$.

Second, after the barycenter has been estimated, BFG-WB evaluates one FGW discrepancy between the estimated barycenter and each treatment-pattern representation distribution. 
Hence, this step requires $N_{\mathrm{FGW}} = L$ additional evaluations.
Consequently, the total number of OT-based discrepancy evaluations required by BFG-WB per training step satisfies
\[
N_{\mathrm{BFG\text{-}WB}}
=
N_{\mathrm{bary}}
+
N_{\mathrm{FGW}}
\leq
L_bL+L
=
(L_b+1)L.
\]
Because $L_b$ is fixed independently of $L$, it follows that
\[
N_{\mathrm{BFG\text{-}WB}}
=
O(L).
\]
Thus, standard pairwise balancing requires a quadratic number
of evaluations in $L$, whereas BFG-WB requires at most a
linear number of evaluations.
\end{proof}
\subsection{Proof of Theorems 1 and 2}
\label{appendix:proof_theorem_upper}
This appendix details the proofs of Theorems 1 and 2, organized into three parts.
First, we summarize the notation, definitions, and key assumptions used throughout the proofs. 
Next, we present supporting lemmas that establish intermediate inequalities connecting counterfactual losses with distributional discrepancies.
Finally, building on these results, we provide the main proofs of Theorems 1 and 2, which demonstrate how barycenter-based balancing yields the stated upper bounds.

\subsubsection{Setup and Definitions}
\label{appendix:assumptions_and_notation}
\paragraph{Notation and setup.}
Let $\phi:\mathcal{X}\to\mathcal{R}\subset\mathbb{R}^{d_r}$ be the representation map, where $\mathcal{R}$ is the representation space.
For each treatment pattern $\boldsymbol{t}\in\mathcal{T}$, we denote by $p(\boldsymbol{x} \mid \boldsymbol{t})$ the covariate distribution conditional on $\boldsymbol{T}=\boldsymbol{t}$, and define $\boldsymbol{R}_{\boldsymbol{t}}$ as the distribution on $\mathcal{R}$ of $\boldsymbol{r}=\phi(\boldsymbol{x})$ when $\boldsymbol{x}\sim p(\boldsymbol{x} \mid \boldsymbol{t})$.
Throughout the theoretical analysis, $\boldsymbol{R}_{\boldsymbol{t}}$ denotes the population distribution on $\mathcal{R}$; in practice, we use its empirical distribution computed from the observed samples in the treatment pattern $\boldsymbol{t}$.
Given $(\boldsymbol{x},\boldsymbol{t})$, we form the predictor input by concatenating the representation and the task embedding $t_w(\boldsymbol{t})$, and define the point predictor of the potential outcome as $\hat{\mu}(\boldsymbol{x},\boldsymbol{t})\;:=\;h\!\left([\phi(\boldsymbol{x}),\,t_w(\boldsymbol{t})]\right)$ where $h$ is the outcome prediction network and $[\cdot,\cdot]$ denotes concatenation.

\begin{definition}
For covariates $\boldsymbol{x}$ and treatment pattern $\boldsymbol{t}$, define the conditional expected squared loss
$
  l(\boldsymbol{x},\boldsymbol{t})
  =
  \int \bigl(Y(\boldsymbol{t})-\hat\mu(\boldsymbol{x}, \boldsymbol{t})\bigr)^2\,p\bigl(Y(\boldsymbol{t})\mid \boldsymbol{x}\bigr)\,dY(\boldsymbol{t}).
$
The expected factual loss under $t$ and the expected counterfactual loss for predicting outcomes under $t$ using covariates drawn from other treatment patterns are:
\begin{equation}
\label{app_eq:factual_loss}
  \epsilon_{\mathrm{F}}^{(\boldsymbol{t})}
  =
  \int_{\mathcal{X}} l(\boldsymbol{x}, \boldsymbol{t})\,p(\boldsymbol{x} \mid \boldsymbol{t})\,p(\boldsymbol{t})\,d\boldsymbol{x},
\end{equation}

\begin{equation}
\label{app_eq:counterfactual_loss}
  \epsilon_{\mathrm{CF}}^{(\boldsymbol{t})}
  =
  \sum_{\boldsymbol{t}^{\prime} \in \mathcal{T},\,\boldsymbol{t}^{\prime}\neq \boldsymbol{t}}
  \int_{\mathcal{X}} l(\boldsymbol{x}, \boldsymbol{t}) \,p(\boldsymbol{x} \mid \boldsymbol{t}^{\prime})\,p(\boldsymbol{t}^{\prime})\,d\boldsymbol{x}.
\end{equation}

\end{definition}

\begin{definition}
\label{def:kantorovich_wasserstein}
Let $P$ and $Q$ be probability measures on $\mathcal{R}\subset\mathbb{R}^{d_r}$ with finite $p$-th moments.
The $p$-Wasserstein distance is defined as
\begin{equation*}
W_p(P,Q)
\coloneqq
\left(
\inf_{\pi\in\Pi(P,Q)}
\int_{\mathcal{R}\times\mathcal{R}}
\|r-r'\|_2^p
\,d\pi(r,r')
\right)^{1/p},
\end{equation*}
where $\Pi(P,Q)$ denotes the set of all couplings on $\mathcal{R}\times\mathcal{R}$ with marginals $P$ and $Q$.
\end{definition}

\begin{definition}
Let $\mathcal{F}$ be a class of real-valued functions on $\mathcal{R}$.
For two distributions $P$ and $Q$ on $\mathcal{R}$, the integral probability metric~(IPM) induced by $\mathcal{F}$ is defined as
\begin{equation*}
  \mathrm{IPM}_{\mathcal{F}}(P,Q)
  =
  \sup_{f\in\mathcal{F}}
  \left|
    \int_{\mathcal{R}} f(r)\,(dP(r)-dQ(r))
  \right|.
\end{equation*}
\end{definition}

\begin{assumption}[Stable Unit Treatment Value Assumption]
\textit{
(1) no interference, meaning that the outcome of one unit is unaffected by the treatment assignments of other units; and (2) consistency of treatment, meaning that the potential outcomes correspond to well-defined and unique treatments~(i.e., $y_i = Y_i(\boldsymbol{t}_i)$).}
\end{assumption}

\begin{assumption}[Ignorability]
\textit{    
For any treatment pattern, the potential outcome is independent of the assigned treatment $\boldsymbol{T}$ given the observed covariates $\boldsymbol{X}$.
Formally, for all $\boldsymbol{t}$, 
$
Y(\boldsymbol{t}) \perp \boldsymbol{T} \mid \boldsymbol{X}.
$
}
\end{assumption}

\begin{assumption}[Overlap]
\textit{
Every unit has a non-zero probability of receiving any treatment pattern given its observed covariates.
Formally, for all $\boldsymbol{t}$ and $\boldsymbol{x}$, 
$
0 < P(\boldsymbol{T}=\boldsymbol{t} \mid \boldsymbol{X}=\boldsymbol{x}) < 1. 
$
}
\end{assumption}

\begin{assumption}[Invertible representation map]\label{assumption:invertible_rep}
The representation map $\phi:\mathcal{X}\to\mathcal{R}\subset\mathbb{R}^{d_r}$ is one-to-one on $\mathcal{X}$.
Without loss of generality, we assume that $\mathcal{R}$ is the image of $\mathcal{X}$ under $\phi$.
Hence, there exists an inverse map $\Psi:\mathcal{R}\to\mathcal{X}$ such that
\begin{align*}
    \Psi(\phi(\boldsymbol{x})) = \boldsymbol{x} \quad (\forall \boldsymbol{x}\in\mathcal{X}).
\end{align*}
\end{assumption}

\begin{assumption}[Lipschitz loss in representation space]\label{assumption:lipschitz_loss}
\textit{
There exists a constant $B_{\phi}>0$ such that, for any $\boldsymbol{t}\in \mathcal{T}$, the function
\begin{equation*}
  g_{\boldsymbol{t}}(\boldsymbol{r})
  :=
  \frac{1}{B_{\phi}}\,
  l\bigl(\Psi(\boldsymbol{r}),\boldsymbol{t}\bigr)
\end{equation*}
belongs to the class of $1$-Lipschitz functions on $\mathcal{R}$.
}
\end{assumption}

\subsubsection{Supporting Lemmas}
\label{appendix:supporting_lemmas}
We present supporting lemmas and clarify their connections to the proofs of the main theorems.
The lemmas can be grouped into three categories.

First, Lemmas~\ref{lem:cf_ipm_wasserstein} and~\ref{lem:cf_bound} relate counterfactual outcome losses to the corresponding factual outcome losses, plus Wasserstein distances between treatment-specific representation distributions.
Lemma~\ref{lem:cf_ipm_wasserstein} serves as a technical preparation for Lemma~\ref{lem:cf_bound}: it shows that the increase in the expected outcome prediction loss for $\boldsymbol{t}'$ under covariates~(and thus representations) from $\boldsymbol{t}$ is controlled by the Wasserstein distance between the representation distributions of $\boldsymbol{t}$ and $\boldsymbol{t}'$.

\begin{lemma}
\label{lem:cf_ipm_wasserstein}
    Under Assumptions 1-5 , for any $\boldsymbol{t},\boldsymbol{t}^{\prime} \in \mathcal{T}$, we have
    \begin{equation*}
        \int_{\mathcal{X}}l(\boldsymbol{x}, \boldsymbol{t}^{\prime})p(\boldsymbol{x}\mid \boldsymbol{t})d\boldsymbol{x} \leq \int_{\mathcal{X}}l(\boldsymbol{x}, \boldsymbol{t}^{\prime})p(\boldsymbol{x}\mid\boldsymbol{t}^{\prime})d\boldsymbol{x} + B_{\phi}W_1(\boldsymbol{R}_{\boldsymbol{t}}, \boldsymbol{R}_{\boldsymbol{t}^{\prime}}).
    \end{equation*}
\end{lemma}
\begin{proof}
    By the definition of $\boldsymbol{R}_{\boldsymbol{t}}$ and the Assumption~\ref{assumption:invertible_rep}, we obtain
    \begin{align*}
        \int_{\mathcal{X}} l(\boldsymbol{x},\boldsymbol{t}^{\prime})\,p(\boldsymbol{x}\mid \boldsymbol{t})\,d\boldsymbol{x}
        &= \int_{\mathcal{R}} l(\Psi(\boldsymbol{r}),\boldsymbol{t}^{\prime})\,dR_{\boldsymbol{t}}(\boldsymbol{r}),\\
        \int_{\mathcal{X}} l(\boldsymbol{x},\boldsymbol{t}^{\prime})\,p(\boldsymbol{x} \mid \boldsymbol{t}^{\prime})\,d\boldsymbol{x}
        &= \int_{\mathcal{R}} l(\Psi(\boldsymbol{r}),\boldsymbol{t}^{\prime})\,dR_{\boldsymbol{t}^{\prime}}(\boldsymbol{r}).
    \end{align*}
    Therefore, we can rewrite the difference of expectations as an integral over $\mathcal{R}$ as follows.
    \begin{align*}
      &
      \int_{\mathcal{X}} l(\boldsymbol{x},\boldsymbol{t}^{\prime})\,p(\boldsymbol{x}\mid \boldsymbol{t})\,d\boldsymbol{x}
      -
      \int_{\mathcal{X}} l(\boldsymbol{x},\boldsymbol{t}^{\prime})\,p(\boldsymbol{x}\mid \boldsymbol{t}^{\prime})\,d\boldsymbol{x}
      \\
      &=
      \int_{\mathcal{R}}
      l\bigl(\Psi(\boldsymbol{r}),\boldsymbol{t}^{\prime}\bigr)\,
      \bigl(dR_{\boldsymbol{t}}(\boldsymbol{r})-dR_{\boldsymbol{t}^{\prime}}(\boldsymbol{r})\bigr)
      \\
      &=
      B_{\phi}
      \int_{\mathcal{R}}
      g_{\boldsymbol{t}^{\prime}}(\boldsymbol{r})\,
      \bigl(dR_{\boldsymbol{t}}(\boldsymbol{r})-dR_{\boldsymbol{t}^{\prime}}(\boldsymbol{r})\bigr)\\
      &\le
      B_{\phi}\left|
      \int_{\mathcal{R}}
      g_{\boldsymbol{t}^{\prime}}(\boldsymbol{r})\,
      \bigl(dR_{\boldsymbol{t}}(\boldsymbol{r})-dR_{\boldsymbol{t}^{\prime}}(\boldsymbol{r})\bigr)
      \right|\\
      &\le
      B_{\phi}\mathrm{IPM}_{\mathcal{F}}(\boldsymbol{R}_{\boldsymbol{t}},\boldsymbol{R}_{\boldsymbol{t}^{\prime}}) \qquad \text{(by Assumption~\ref{assumption:lipschitz_loss})},
    \end{align*}
    where $\mathcal{F}=\mathrm{Lip}_1(\mathcal{R})$ denotes the set of all $1$-Lipschitz functions on $\mathcal{R}$.
    By the Kantorovich-Rubinstein duality~\citep{villani2008optimal}, the IPM over $1$-Lipschitz functions equals the 1-Wasserstein distance,
    \begin{equation*}
      \mathrm{IPM}_{\mathcal{F}}(\boldsymbol{R}_{\boldsymbol{t}},\boldsymbol{R}_{\boldsymbol{t}^{\prime}}) = W_1(\boldsymbol{R}_{\boldsymbol{t}},\boldsymbol{R}_{\boldsymbol{t}^{\prime}}).
    \end{equation*}
    Accordingly, we have
    \begin{equation*}
      \int_{\mathcal{X}} l(\boldsymbol{x},\boldsymbol{t}^{\prime})\,p(\boldsymbol{x}\mid \boldsymbol{t})\,d\boldsymbol{x}
      \le
      \int_{\mathcal{X}} l(\boldsymbol{x},\boldsymbol{t}^{\prime})\,p(\boldsymbol{x}\mid \boldsymbol{t}^{\prime})\,d\boldsymbol{x}
      +
      B_{\phi} W_1(\boldsymbol{R}_{\boldsymbol{t}},\boldsymbol{R}_{\boldsymbol{t}^{\prime}}).
    \end{equation*}

\end{proof}

Building on this result, Lemma~\ref{lem:cf_bound} upper-bounds the counterfactual outcome loss for treatment pattern $\boldsymbol{t}$ by its factual outcome loss plus Wasserstein discrepancies to other treatment patterns.
% counterfactual upper bound
\begin{lemma}
\label{lem:cf_bound}
    Under the Assumptions 1-5, for any treatment pattern $\boldsymbol{t} \in \mathcal{T}$, the expected counterfactual loss satisfies
    \begin{equation*}
    \epsilon_{\mathrm{CF}}^{(\boldsymbol{t})}
    \le
    \left(\frac{1}{p(\boldsymbol{t})}-1\right)\epsilon_{\mathrm{F}}^{(\boldsymbol{t})}
    \;+\;
    B_{\phi}
    \sum_{\substack{\boldsymbol{t}^{\prime}\ne \boldsymbol{t}}}
    \,W_1\!\left(\boldsymbol{R}_{\boldsymbol{t}},\boldsymbol{R}_{\boldsymbol{t}^{\prime}}\right).
    \end{equation*}
\end{lemma}
\begin{proof}
Applying Lemma~\ref{lem:cf_ipm_wasserstein} to each $\boldsymbol{t}^{\prime}$ yields
\begin{align}
\label{eq:lem_prepare}
\int_{\mathcal{X}}
l(\boldsymbol{x}, \boldsymbol{t})\,p(\boldsymbol{x} \mid \boldsymbol{t}^{\prime})\,d\boldsymbol{x}
\le
\int_{\mathcal{X}}
l(\boldsymbol{x} , \boldsymbol{t})\,p(\boldsymbol{x}\mid \boldsymbol{t})\,d\boldsymbol{x}
+
B_{\phi}\,W_1\!\left(\boldsymbol{R}_{\boldsymbol{t}},\boldsymbol{R}_{\boldsymbol{t}^{\prime}}\right).
\end{align}

Multiplying both sides of \eqref{eq:lem_prepare} by $p(\boldsymbol{t}^{\prime})$ and summing over $\boldsymbol{t}^{\prime}\ne \boldsymbol{t}$, we obtain
\begin{align*}
\begin{split}
\label{eq:counter_factual_upper_proof}
\epsilon_{\mathrm{CF}}^{(\boldsymbol{t})}
&\le
\sum_{\boldsymbol{t}^{\prime}\ne \boldsymbol{t}}
p(\boldsymbol{t}^{\prime})\int_{\mathcal{X}}l(\boldsymbol{x},\boldsymbol{t})\,p(\boldsymbol{x}\mid \boldsymbol{t})\,d\boldsymbol{x}
+
B_{\phi}\sum_{\boldsymbol{t}^{\prime}\ne \boldsymbol{t}}p(\boldsymbol{t}^{\prime})W_1\!\left(\boldsymbol{R}_{\boldsymbol{t}},\boldsymbol{R}_{\boldsymbol{t}^{\prime}}\right) \\
&=
\left(\sum_{\boldsymbol{t}^{\prime}\ne \boldsymbol{t}}p(\boldsymbol{t}^{\prime})\right)\int_{\mathcal{X}}l(\boldsymbol{x},\boldsymbol{t})\,p(\boldsymbol{x}\mid \boldsymbol{t})\,d\boldsymbol{x}
+
B_{\phi}\sum_{\boldsymbol{t}^{\prime}\ne \boldsymbol{t}}p(\boldsymbol{t}^{\prime})W_1\!\left(\boldsymbol{R}_{\boldsymbol{t}},\boldsymbol{R}_{\boldsymbol{t}^{\prime}}\right) \\
&=
(1-p(\boldsymbol{t}))\int_{\mathcal{X}}l(\boldsymbol{x}, \boldsymbol{t})\,p(\boldsymbol{x}\mid \boldsymbol{t})\,d\boldsymbol{x}
+
B_{\phi}\sum_{\boldsymbol{t}^{\prime}\ne \boldsymbol{t}}p(\boldsymbol{t}^{\prime})W_1\!\left(\boldsymbol{R}_{\boldsymbol{t}},\boldsymbol{R}_{\boldsymbol{t}^{\prime}}\right) \\
&=\left(\frac{1}{p(\boldsymbol{t})}-1\right)\epsilon_{\mathrm{F}}^{(\boldsymbol{t})}
+
B_{\phi}\sum_{\boldsymbol{t}^{\prime}\ne \boldsymbol{t}} p(\boldsymbol{t}^{\prime}) W_1\!\left(\boldsymbol{R}_{\boldsymbol{t}},\boldsymbol{R}_{\boldsymbol{t}^{\prime}}\right) \qquad \text{(by Definition in Equation~\eqref{app_eq:factual_loss})} \\
&\le \left(\frac{1}{p(\boldsymbol{t})}-1\right)\epsilon_{\mathrm{F}}^{(\boldsymbol{t})}
+
B_{\phi}\sum_{\boldsymbol{t}^{\prime}\ne \boldsymbol{t}} W_1\!\left(\boldsymbol{R}_{\boldsymbol{t}},\boldsymbol{R}_{\boldsymbol{t}^{\prime}}\right).
\end{split}
\end{align*}
where the last inequality holds because $p(\boldsymbol{t}^{\prime}) \le 1$ and $W_1(\cdot, \cdot) \ge 0$.
\end{proof}

% Pairwise Wasserstein discrepancies are upper-bounded by barycentric discrepancies
Second, Lemma~\ref{lem:pairwise_to_barycenter_bound} provides an upper bound on aggregate discrepancies across multiple treatment patterns in terms of distances to a common Wasserstein barycenter.
This bound demonstrates that the sum of all pairwise discrepancies can be controlled by reducing the discrepancies between each representation distribution and the Wasserstein barycenter.
This is exactly what our barycenter-based regularization term $\mathcal{L}_{\phi}$ enforces.
\begin{lemma}\label{lem:pairwise_to_barycenter_bound}
    Let $\mathcal{T}$ denote the set of all treatment patterns with $|\mathcal{T}|=2^{K}$.
    For any $p\ge 1$, the following inequality holds:
    \begin{equation*}
    \sum_{\{\boldsymbol{t},\boldsymbol{t}'\}\in \binom{\mathcal{T}}{2}}
    W_p\!\left(\boldsymbol{R}_{\boldsymbol{t}}, \boldsymbol{R}_{\boldsymbol{t}'}\right)
    \le
    \left(2^{K}-1\right)
    \sum_{\boldsymbol{t}\in \mathcal{T}}
    W_p\!\left(\boldsymbol{R}_b^{\ast}, \boldsymbol{R}_{\boldsymbol{t}}\right),
    \end{equation*}
    where $\boldsymbol{R}_b^{\ast}$ is a Wasserstein barycenter of $\{\boldsymbol{R}_{\boldsymbol{t}}\}_{\boldsymbol{t}\in\mathcal{T}}$.
\end{lemma}
\begin{proof}
    By the triangle inequality for the $p$-Wasserstein distance defined in Definition~\ref{def:kantorovich_wasserstein}, for any $\boldsymbol{t},\boldsymbol{t}'\in \mathcal{T}$,
    \begin{equation*}
    W_p\!\left(\boldsymbol{R}_{\boldsymbol{t}}, \boldsymbol{R}_{\boldsymbol{t}'}\right)
    \le
    W_p\!\left(\boldsymbol{R}_{\boldsymbol{t}}, \boldsymbol{R}_b^{\ast}\right)
    +
    W_p\!\left(\boldsymbol{R}_b^{\ast}, \boldsymbol{R}_{\boldsymbol{t}'}\right).
    \end{equation*}
    
    Summing the above inequality over all \emph{ordered} pairs $(\boldsymbol{t},\boldsymbol{t}')$ with $\boldsymbol{t}\neq \boldsymbol{t}'$ yields
    \begin{align*}
    \sum_{\boldsymbol{t}\neq \boldsymbol{t}'}
    W_p\!\left(\boldsymbol{R}_{\boldsymbol{t}}, \boldsymbol{R}_{\boldsymbol{t}'}\right)
    &\le
    \sum_{\boldsymbol{t}\neq \boldsymbol{t}'}
    \left\{
    W_p\!\left(\boldsymbol{R}_{\boldsymbol{t}}, \boldsymbol{R}_b^{\ast}\right)
    +
    W_p\!\left(\boldsymbol{R}_b^{\ast}, \boldsymbol{R}_{\boldsymbol{t}'}\right)
    \right\} \\
    &=
    (|\mathcal{T}|-1)\sum_{\boldsymbol{t}\in\mathcal{T}} W_p\!\left(\boldsymbol{R}_{\boldsymbol{t}}, \boldsymbol{R}_b^{\ast}\right)
    +
    (|\mathcal{T}|-1)\sum_{\boldsymbol{t}\in\mathcal{T}} W_p\!\left(\boldsymbol{R}_b^{\ast}, \boldsymbol{R}_{\boldsymbol{t}}\right),
    \end{align*}
    since each $\boldsymbol{t}\in\mathcal{T}$ appears exactly $|\mathcal{T}|-1$ times as the first element of an ordered pair.

    Using the symmetry of $W_p$, we have
    $\sum_{\boldsymbol{t}\in\mathcal{T}} W_p(\boldsymbol{R}_{\boldsymbol{t}}, \boldsymbol{R}_b^{\ast})
    =
    \sum_{\boldsymbol{t}\in\mathcal{T}} W_p(\boldsymbol{R}_b^{\ast}, \boldsymbol{R}_{\boldsymbol{t}})$,
    and thus
    \begin{equation*}
    \sum_{\boldsymbol{t}\neq \boldsymbol{t}'}
    W_p\!\left(\boldsymbol{R}_{\boldsymbol{t}}, \boldsymbol{R}_{\boldsymbol{t}'}\right)
    \le
    2(|\mathcal{T}|-1)\sum_{\boldsymbol{t}\in\mathcal{T}} W_p\!\left(\boldsymbol{R}_b^{\ast}, \boldsymbol{R}_{\boldsymbol{t}}\right).
    \end{equation*}

    Finally, noting that
    $\sum_{\boldsymbol{t}\neq \boldsymbol{t}'} W_p(\boldsymbol{R}_{\boldsymbol{t}},\boldsymbol{R}_{\boldsymbol{t}'})
    =2\sum_{\{\boldsymbol{t},\boldsymbol{t}'\}\in\binom{\mathcal{T}}{2}} W_p(\boldsymbol{R}_{\boldsymbol{t}},\boldsymbol{R}_{\boldsymbol{t}'})$,
    we obtain
    \begin{equation*}
    \sum_{\{\boldsymbol{t},\boldsymbol{t}'\}\in \binom{\mathcal{T}}{2}}
    W_p\!\left(\boldsymbol{R}_{\boldsymbol{t}}, \boldsymbol{R}_{\boldsymbol{t}'}\right)
    \le
    (|\mathcal{T}|-1)
    \sum_{\boldsymbol{t}\in \mathcal{T}}
    W_p\!\left(\boldsymbol{R}_b^{\ast}, \boldsymbol{R}_{\boldsymbol{t}}\right)
    =
    \left(2^{K}-1\right)
    \sum_{\boldsymbol{t}\in \mathcal{T}}
    W_p\!\left(\boldsymbol{R}_b^{\ast}, \boldsymbol{R}_{\boldsymbol{t}}\right),
    \end{equation*}
    which completes the proof.
    \end{proof}
Although this lemma holds for any anchor distribution, using the Wasserstein barycenter can yield a smaller right-hand side, leading to a tighter anchor-based bound.

% relation of wasserstein and fgw
Third, Lemma~\ref{lem:w1_upper_by_fgw} provides an upper bound that relates the above Wasserstein discrepancies to the FGW discrepancy adopted in our objective, thereby allowing us to control the Wasserstein terms via FGW.
Whereas replacing the Wasserstein terms with an FGW-based upper bound may loosen the theoretical bound, introducing FGW is crucial in practice because it encourages the preservation of local proximity structures in the representation space.
\begin{lemma}
\label{lem:w1_upper_by_fgw}
    Let $0<\eta\le 1$ and let $F(\cdot,\cdot)$ be the FGW discrepancy defined in Equation~\eqref{eq:fgw_to_barycenter}.
    Then, for any treatment pattern $\boldsymbol{t}\in\mathcal{T}$,
    \begin{equation*}
    W_1\!\left(\boldsymbol{R}_{\boldsymbol{t}}, \boldsymbol{R}_{b}^{\ast}\right)
    \le
    \frac{1}{\eta}\,
    F\!\left(\boldsymbol{R}_{\boldsymbol{t}}, \boldsymbol{R}_{b}^{\ast}\right).
    \end{equation*}
\end{lemma}

\begin{proof}
For any coupling $\pi\in\Pi(\boldsymbol{R}_{\boldsymbol{t}},\boldsymbol{R}_b^*)$, write the FGW objective as $ J(\pi) = \eta A(\pi) + (1-\eta)B(\pi)$, where $A(\pi)$ is the feature-based transport cost and $B(\pi)\geq 0$ is the structural term.
By Definition~\ref{def:kantorovich_wasserstein},
\[
A(\pi)
\geq
\inf_{\widetilde{\pi}
\in\Pi(\boldsymbol{R}_{\boldsymbol{t}},\boldsymbol{R}_b^*)}
A(\widetilde{\pi})
=
W_1(R_{\boldsymbol{t}},R_b^*).
\]
Therefore, $J(\pi) \geq \eta W_1(\boldsymbol{R}_{\boldsymbol{t}},\boldsymbol{R}_b^*)$.
Taking the infimum over $\pi\in\Pi(\boldsymbol{R}_{\boldsymbol{t}},\boldsymbol{R}_b^*)$ yields $F(\boldsymbol{R}_{\boldsymbol{t}},\boldsymbol{R}_b^*) \geq \eta W_1(\boldsymbol{R}_{\boldsymbol{t}},\boldsymbol{R}_b^*)$ ,which proves the result.
\end{proof}

Combining Lemmas~\ref{lem:cf_ipm_wasserstein} - \ref{lem:w1_upper_by_fgw}, we can upper-bound the counterfactual outcome losses terms in Theorems 1 and 2 by the corresponding factual outcome losses plus the barycenter-based regularizer~(via Lemma~\ref{lem:pairwise_to_barycenter_bound}) and the FGW discrepancy term in our objective~(via Lemma~\ref{lem:w1_upper_by_fgw}).
We then obtain the stated bounds by substituting these inequalities into the decompositions used in the proofs of Theorems 1 and 2.

\subsubsection{Main Proofs: Theorems 1 and 2}
Using the above lemmas, we prove Theorems 1 and 2.

\begin{thm}[Upper bound for CASE]\label{thm:case_upper}
Suppose that Assumptions 1–3 and the auxiliary conditions in Appendix~\ref{appendix:assumptions_and_notation} hold~(in particular, the expected squared loss $l(\boldsymbol{x},\boldsymbol{t})$ is $B_{\phi}$-Lipschitz in the representation space), and that the weights in Equation~\eqref{eq:bfgwb_overview} are uniform, i.e., $w_{\boldsymbol{t}}=2^{-K}$. 
Then, for any $k \in \{1,\ldots,K\}$,
\begin{equation*}
\epsilon_{\mathrm{CASE}}(k)
\le
2\left(
\frac{1}{p(\boldsymbol{t}_{+k})}\,\epsilon_{\mathrm{F}}^{(\boldsymbol{t}_{+k})}
+
\frac{1}{p(\boldsymbol{0})}\,\epsilon_{\mathrm{F}}^{(\boldsymbol{0})}
+
\frac{2^{2K}}{\eta}\,B_{\phi}\mathcal{L}_{\phi}
\right).
\end{equation*}
\end{thm}
\begin{proof}
    By definition in Equation~\eqref{eq:def_case_error}, we rewrite
    \begin{align*}
    \epsilon_{\mathrm{CASE}}(k)
    &= \int_{\mathcal{X}} \Bigl[ \bigl\{\hat{\mu}(\boldsymbol{x},\boldsymbol{t}_{+k})-\mu(\boldsymbol{x},\boldsymbol{t}_{+k})\bigr\} - \bigl\{\hat{\mu}(\boldsymbol{x}, \boldsymbol{0})-\mu(\boldsymbol{x}, \boldsymbol{0})\bigr\} \Bigr]^2
    \,p(\boldsymbol{x})\,d\boldsymbol{x}.
    \end{align*}
    Therefore, we have
    \begin{align*}
    \epsilon_{\mathrm{CASE}}(k)
    &\le 2 \int_{\mathcal{X}} \Bigl[ \bigl\{ \hat{\mu}(\boldsymbol{x},\boldsymbol{t}_{+k}) - \mu(\boldsymbol{x},\boldsymbol{t}_{+k}) \bigr\}^2
        + 
        \bigl\{\hat{\mu}(\boldsymbol{x}, \boldsymbol{0}) - \mu(\boldsymbol{x}, \boldsymbol{0}) \bigr\}^2 \Bigr] \, p(\boldsymbol{x})\, d\boldsymbol{x} \\
    &= 2 \sum_{\boldsymbol{t}^{\prime} \in \mathcal{T}} \int_{\mathcal{X}} \bigl\{ \hat{\mu}(\boldsymbol{x}, \boldsymbol{t}_{+k}) - \mu(\boldsymbol{x}, \boldsymbol{t}_{+k}) \bigr\}^2 \, p(\boldsymbol{x} \mid \boldsymbol{t}^{\prime})\, p(\boldsymbol{t}^{\prime}) d\boldsymbol{x}
    +
    2 \sum_{\boldsymbol{t}^{\prime} \in \mathcal{T}} \int_{\mathcal{X}}  \bigl\{ \hat{\mu}(\boldsymbol{x},\boldsymbol{0}) - \mu(\boldsymbol{x},\boldsymbol{0}) \bigr\}^2 \, p(\boldsymbol{x} \mid \boldsymbol{t}^{\prime})\, p(\boldsymbol{t}^{\prime}) d\boldsymbol{x}  \\
    &\le
    2\bigl(\epsilon_{\mathrm{F}}^{(\boldsymbol{t}_{+k})}+\epsilon_{\mathrm{CF}}^{(\boldsymbol{t}_{+k})}\bigr)
    +
    2\bigl(\epsilon_{\mathrm{F}}^{(\boldsymbol{0})}+\epsilon_{\mathrm{CF}}^{(\boldsymbol{0})}\bigr)  \qquad \text{(by the Definition of Equations~\eqref{app_eq:factual_loss} and \eqref{app_eq:counterfactual_loss})} \\
    &\le   2\left( \frac{1}{p(\boldsymbol{t}_{+k})}\epsilon_{\mathrm{F}}^{(\boldsymbol{t}_{+k})} + \frac{1}{p(\boldsymbol{0})}\epsilon_{\mathrm{F}}^{(\boldsymbol{0})} \right)
    + 2B_{\phi} \left( \sum_{\boldsymbol{t}^{\prime}\ne \boldsymbol{t}_{+k}} W_1(\boldsymbol{R}_{\boldsymbol{t}_{+k}},\boldsymbol{R}_{\boldsymbol{t}^{\prime}}) + \sum_{\boldsymbol{t}^{\prime}\ne \boldsymbol{0}}  W_1(\boldsymbol{R}_{\boldsymbol{0}},\boldsymbol{R}_{\boldsymbol{t}^{\prime}})\right)  \qquad \text{(by Lemma~\ref{lem:cf_bound})} ,
    \end{align*}
    where we use $(\hat{\mu}(\boldsymbol{x},\boldsymbol{t})-\mu(\boldsymbol{x},\boldsymbol{t}))^{2}\le l(\boldsymbol{x},\boldsymbol{t})$, which follows from the law of total variance.

    To connect the remaining Wasserstein terms to the BFG-WB regularization, we upper bound the pairwise sums by barycenter-based discrepancies.
    Let $\boldsymbol{R}_b^\ast$ denote the Wasserstein barycenter over $\{\boldsymbol{R}_{\boldsymbol{t}}\}_{\boldsymbol{t}\in\mathcal{T}}$.
    For any fixed $\boldsymbol{t} \in\mathcal{T}$, by the triangle inequality of $W_1$,
    \begin{align}
    \label{eq:tri_to_bary}
    \begin{split}
    \sum_{\boldsymbol{t}^{\prime}\ne \boldsymbol{t}} W_1(\boldsymbol{R}_{\boldsymbol{t}},\boldsymbol{R}_{\boldsymbol{t}^{\prime}})
    &\le \sum_{\boldsymbol{t}^{\prime}\ne \boldsymbol{t}}\Bigl\{W_1(\boldsymbol{R}_{\boldsymbol{t}}, \boldsymbol{R}_b^\ast)+W_1(\boldsymbol{R}_{\boldsymbol{t}^{\prime}},\boldsymbol{R}_b^\ast)\Bigr\} \\
    &= (2^K-1)\,W_1(\boldsymbol{R}_{\boldsymbol{t}},\boldsymbol{R}_b^\ast)+\sum_{\boldsymbol{t}^{\prime}\ne \boldsymbol{t}}W_1(\boldsymbol{R}_{\boldsymbol{t}^{\prime}},\boldsymbol{R}_b^\ast) \\
    &= (2^K-2)\,W_1(\boldsymbol{R}_{\boldsymbol{t}},\boldsymbol{R}_b^\ast)+\sum_{\boldsymbol{t}^{\prime}\in\mathcal{T}}W_1(\boldsymbol{R}_{\boldsymbol{t}^{\prime}},\boldsymbol{R}_b^\ast).
    \end{split}
    \end{align}
    Applying~\eqref{eq:tri_to_bary} to $\boldsymbol{t}=\boldsymbol{t}_{+k}$ and $\boldsymbol{t}=\boldsymbol{0}$ and summing them yields
    \begin{align*}
    \sum_{\boldsymbol{t}^{\prime}\ne \boldsymbol{t}_{+k}} W_1(\boldsymbol{R}_{\boldsymbol{t}_{+k}},\boldsymbol{R}_{\boldsymbol{t}^{\prime}})
    +\sum_{\boldsymbol{t}^{\prime}\ne \boldsymbol{0}} W_1(\boldsymbol{R}_{\boldsymbol{0}},\boldsymbol{R}_{\boldsymbol{t}^{\prime}})
    &\le (2^K-2)\Bigl\{W_1(\boldsymbol{R}_{\boldsymbol{t}_{+k}},\boldsymbol{R}_b^\ast)+W_1(\boldsymbol{R}_{\boldsymbol{0}},\boldsymbol{R}_b^\ast)\Bigr\}
    +2\sum_{\boldsymbol{t}^{\prime} \in\mathcal{T}}W_1(\boldsymbol{R}_{\boldsymbol{t}^{\prime}}, \boldsymbol{R}_b^\ast) \\
    &\le 2^K\sum_{\boldsymbol{t}\in\mathcal{T}} W_1(\boldsymbol{R}_{\boldsymbol{t}},\boldsymbol{R}_b^\ast)\\
    &\le \frac{2^K}{\eta}\sum_{\boldsymbol{t}\in\mathcal{T}}F(\boldsymbol{R}_{\boldsymbol{t}}, \boldsymbol{R}_b^\ast) \\
    &\le \frac{2^{2K}}{\eta}\mathcal{L} _\phi.
    \end{align*}
\end{proof}

\begin{thm}[Upper Bound for CAIE]\label{thm:caie_bound}
Under the same assumptions as Theorem~\ref{thm:case_upper}, consider an interaction set $S \subseteq \{1,\ldots,K\}$ with $|S|\ge 2$.
Let $(a_{\boldsymbol{t}})_{\boldsymbol{t}\in\mathcal{T}}$  be a constant vector reflecting the combinatorial structure of $S$~($a_{\boldsymbol{t}}\in\{-1,0,1\}$). 
The estimation error $\epsilon_{\mathrm{CAIE}}(S)$ is bounded by:

\begin{equation*}\label{eq:caie_upper_bound}
\epsilon_{\mathrm{CAIE}}(S)
\;\le\;
\left(\sum_{\boldsymbol{t}\in\mathcal{T}} a_{\boldsymbol{t}}^2\right)
\left\{
\sum_{\boldsymbol{t}\in\mathcal{T}}
\frac{1}{p(\boldsymbol{t})}\,\epsilon_{\mathrm{F}}^{(\boldsymbol{t})}
\;+\;
\frac{2^{K+1}}{\eta}\,B_{\phi}\,(2^{K}-1)\,\mathcal{L}_{\phi}
\right\}.
\end{equation*}
\end{thm}
\begin{proof}
\begin{align*}
        \epsilon_{\mathrm{CAIE}}(S)
        &= \int_{\mathcal{X}} \left\{\sum_{\boldsymbol{t}} a_{\boldsymbol{t}} \bigl( \hat{\mu}(\boldsymbol{x}, \boldsymbol{t}) - \mu(\boldsymbol{x}, \boldsymbol{t}) \bigr) \right\}^2 p(\boldsymbol{x})\, d\boldsymbol{x} \qquad \text{(by Definition in Equation~\eqref{eq:def_caie_error})}\\
        &\le  \left( \sum_{\boldsymbol{t}} a_{\boldsymbol{t}}^2 \right) \sum_{\boldsymbol{t}} \int_{\mathcal{X}} \bigl(\hat{\mu}(\boldsymbol{x}, \boldsymbol{t}) - \mu(\boldsymbol{x}, \boldsymbol{t}) \bigr)^2 p(\boldsymbol{x})\, d\boldsymbol{x} \qquad \text{(by Cauchy--Schwarz)}  \\
        &=  \left(\sum_{\boldsymbol{t}} a_{\boldsymbol{t}}^2 \right)  \sum_{\boldsymbol{t}}\sum_{\boldsymbol{t}^{\prime}} \int_{\mathcal{X}} \bigl(\hat{\mu}(\boldsymbol{x}, \boldsymbol{t}) - \mu(\boldsymbol{x}, \boldsymbol{t}) \bigr)^2\, p(\boldsymbol{x} \mid \boldsymbol{t}^{\prime})\, p(\boldsymbol{t}^{\prime})\, d\boldsymbol{x}  \\
         &=  \left( \sum_{\boldsymbol{t}} a_{\boldsymbol{t}}^2 \right) \sum_{\boldsymbol{t}} \left\{ \int_{\mathcal{X}} \bigl(\hat{\mu}(\boldsymbol{x}, \boldsymbol{t}) - \mu(\boldsymbol{x}, \boldsymbol{t}) \bigr)^2\, p(\boldsymbol{x} \mid \boldsymbol{t})\, p(\boldsymbol{t})\, d\boldsymbol{x} +
        \sum_{\boldsymbol{t}^{\prime} \neq \boldsymbol{t}} \int_{\mathcal{X}} \bigl(\hat{\mu}(\boldsymbol{x}, \boldsymbol{t}) - \mu(\boldsymbol{x}, \boldsymbol{t}) \bigr)^2\, p(\boldsymbol{x} \mid \boldsymbol{t}^{\prime})\, p(\boldsymbol{t}^{\prime})\, d\boldsymbol{x} \right\} \\
        &\le  \left( \sum_{\boldsymbol{t}} a_{\boldsymbol{t}}^2 \right) \sum_{\boldsymbol{t}} \left\{ \int_{\mathcal{X}} l(\boldsymbol{x}, \boldsymbol{t})\, p(\boldsymbol{x} \mid \boldsymbol{t})\, p(\boldsymbol{t})\, d\boldsymbol{x} +
        \sum_{\boldsymbol{t}^{\prime} \neq \boldsymbol{t}} \int_{\mathcal{X}} l(\boldsymbol{x},\boldsymbol{t})\, p(\boldsymbol{x} \mid \boldsymbol{t}^{\prime})\, p(\boldsymbol{t}^{\prime})\, d\boldsymbol{x} \right\} \\
        &= \left( \sum_{\boldsymbol{t}} a_{\boldsymbol{t}}^2 \right) \sum_{\boldsymbol{t}}\left( \epsilon_{\mathrm{F}}^{(\boldsymbol{t})} + \epsilon_{\mathrm{CF}}^{(\boldsymbol{t})} \right)  \qquad \text{(by Definition in Equations~\eqref{app_eq:factual_loss} and \eqref{app_eq:counterfactual_loss})}\\
        &\le \left(\sum_{\boldsymbol{t}} a_{\boldsymbol{t}}^2 \right) \sum_{\boldsymbol{t}}\left\{\epsilon_{\mathrm{F}}^{(\boldsymbol{t})} + \left( \frac{1}{p(\boldsymbol{t})} - 1 \right) \epsilon_{\mathrm{F}}^{(\boldsymbol{t})}
         + {B}_{\phi}\,  \sum_{ \boldsymbol{t^{\prime}} \neq \boldsymbol{t}}  W_1(\boldsymbol{R}_{\boldsymbol{t}},\boldsymbol{R}_{\boldsymbol{\boldsymbol{t}^{\prime}}})  \right\} \\
\end{align*}

Note that $\sum_{\boldsymbol{t}}\sum_{\boldsymbol{t}'\neq \boldsymbol{t}} W_1(\boldsymbol{R}_{\boldsymbol{t}},\boldsymbol{R}_{\boldsymbol{t}^{\prime}})$
counts each unordered pair $\{\boldsymbol{t},\boldsymbol{t}^{\prime}\}$ twice.
Hence,
\[
\sum_{\boldsymbol{t}}\sum_{\boldsymbol{t}^{\prime} \neq \boldsymbol{t}} W_1(\boldsymbol{R}_{\boldsymbol{t}},\boldsymbol{R}_{\boldsymbol{t}^{\prime}})
= 2 \sum_{\{\boldsymbol{t},\boldsymbol{t}^{\prime}\}\in{\binom{\mathcal{T}}{2}}} W_1(\boldsymbol{R}_{\boldsymbol{t}},\boldsymbol{R}_{\boldsymbol{t}^{\prime}}).
\]

By substituting this identity into the previous inequality, we have
\begin{align*}
\epsilon_{\mathrm{CAIE}}(S)
     &= \left( \sum_{\boldsymbol{t}} a_{\boldsymbol{t}}^2 \right)\left\{ \sum_{\boldsymbol{t}} \frac{1}{p(\boldsymbol{t})}\,\epsilon_{\mathrm{F}}^{(\boldsymbol{t})} + 2B_{\phi} \sum_{\{\boldsymbol{t},\boldsymbol{t}^{\prime}\}\in{\binom{\mathcal{T}}{2}}} W_1(\boldsymbol{R}_{\boldsymbol{t}},\boldsymbol{R}_{\boldsymbol{t}^{\prime}})  \right\} \\
        &\le \left( \sum_{\boldsymbol{t}} a_{\boldsymbol{t}}^2 \right)\left\{ \sum_{\boldsymbol{t}} \frac{1}{p(\boldsymbol{t})}\,\epsilon_{\mathrm{F}}^{(\boldsymbol{t})} 
        + 2B_{\phi}(2^{K}-1)\sum_{\boldsymbol{t}} W_1(\boldsymbol{R}_{\boldsymbol{t}},\boldsymbol{R}_b^\ast)  \right\} \quad \text{(by Lemma~\ref{lem:pairwise_to_barycenter_bound})}\\
        &\le \left( \sum_{\boldsymbol{t}} a_{\boldsymbol{t}}^2 \right)\left\{ \sum_{\boldsymbol{t}} \frac{1}{p(\boldsymbol{t})}\,\epsilon_{\mathrm{F}}^{(\boldsymbol{t})} 
        + \frac{2}{\eta}B_{\phi}(2^{K}-1)\sum_{\boldsymbol{t}} F(\boldsymbol{R}_{\boldsymbol{t}},\boldsymbol{R}_b^\ast)  \right\} \quad \text{(by Lemma~\ref{lem:w1_upper_by_fgw})}\\
        &= \left( \sum_{\boldsymbol{t}} a_{\boldsymbol{t}}^2 \right) \left\{ \sum_{\boldsymbol{t}} \frac{1}{p(\boldsymbol{t})}\, \epsilon_{\mathrm{F}}^{(\boldsymbol{t})} 
         + \frac{2^{K+1}}{\eta}\,{B}_{\phi}\, (2^{K}-1)\, \mathcal{L}_{\phi}  \right\}.
\end{align*}

\end{proof}

Although the two theorems could be subsumed under a general linear-contrast bound, presenting them separately clarifies how BFG-WB controls the estimation errors of both single and interaction treatment effects. 
Together, these bounds provide a principled decomposition that motivates the training objective in Equation~\eqref{eq:total_loss} by linking factual prediction errors and the BFG-WB discrepancy to the estimation errors of both estimands.

It is important to note the connection between these theoretical bounds and our empirical design. 
Assumptions 4 and 5 are idealized conditions; they follow the standard setting adopted in representation-based causal inference~\cite{shalit2017estimating,wang2025proximity}. 
Assumption 5 concerns the smoothness of the loss over the representation space, and is more plausible when nearby representations correspond to units with similar outcomes. 
By penalizing distortions of within-distribution neighborhoods, the proposed BFG-WB objective is designed to encourage representations of this kind, which we examine empirically in Appendix D.4

\section{Data Generating Process}
\label{appendix:dgp}

We describe the data-generating process~(DGP) for the simulation datasets used in Section~\ref{main:simulation} to evaluate the ability of the models to estimate CASE and CAIE accurately and stably under selection bias.
We generate two datasets under covariate-dependent treatment assignment and vary only whether the outcome function contains interaction effects. 
Simulation 1 includes interaction effects among treatments in the outcome generation, which serves as the primary benchmark for CASE and CAIE estimation.
Simulation 2 excludes interaction effects to assess the model's robustness and its ability to avoid inducing false interaction effects.

In two scenarios, the simulation data follow the same functional form for covariate generation, treatment assignment, and outcome generation.
The number of treatments $K$ is fixed at 3, and the sample size $N$ is set to 50,000.
We introduce the indicator variables $\boldsymbol{H}$ to reflect realistic conditions where treatment assignment depends on covariates~\citep{ascarza2018retention,djulbegovic2014physicians}.
\begin{gather*}
x_{i,n}^{(j)} \sim N(c_n^{(j)}, 1^2), \quad 
x_{i,u}^{(j)} \sim U(-1,1), \quad 
\boldsymbol{x}_{i} = (\boldsymbol{x}_{i,n}, \boldsymbol{x}_{i,u}),\\
H_{i}^{(v)} = \mathbb{I}\left(\,x_{i}^{(v)} + x_{i}^{(v+1)} > 1\,\right), \quad v=1,2,3,\\
\boldsymbol{H}_{i} = (H_{i}^{(1)}, H_{i}^{(2)}, H_{i}^{(3)}),\\
t_{i}^{(k)} \sim \mathrm{Bern} \left( \sigma \left( \boldsymbol{w}_{t_k}^\top \boldsymbol{x}_{i}+ \boldsymbol{w}_{H,k}^{\top}\boldsymbol{H}_{i} - 1 \right) \right), \quad k=1,2,3,\\
\boldsymbol{t_i} = (t_{i}^{(1)},t_{i}^{(2)},t_{i}^{(3)}),\\
Y_i \sim \mathcal{N}\!\left(f\!\left(\boldsymbol{x}_{i},\boldsymbol{t}_i,l\right),\,1^2\right),\quad l \in \{0,1\},
\end{gather*}
where $j \in \{1,\dots,15\}$ and $c_n^{(j)}$ are drawn from the uniform distribution $U(-1,1)$. 
$\mathbb{I}(\cdot)$ is the indicator function, $\sigma(\cdot)$ is the sigmoid function defined as $\sigma(x) = 1 / (1 + \exp(-x))$.
$\boldsymbol{w}_{t_k}$ and $\boldsymbol{w}_{H,k}$ are the weight vectors whose elements are independently drawn from the uniform distribution $U(-0.5,0.5)$. 
The outcome-generating function $f$ is defined as follows:
\begin{align*}
f(\boldsymbol{x}_{i}, \boldsymbol{t}_i, l) 
&= \boldsymbol{w}_x^\top \boldsymbol{x}_{i}
 + (x_{i}^{(1)}+1)t_{i}^{(1)} 
 + 1.2(x_{i}^{(2)}+1)t_{i}^{(2)} \\
&\quad + 0.8(x_{i}^{(3)}+1)t_{i}^{(3)} 
+ l\Big\{ (x_{i}^{(4)}+0.5)t_{i}^{(1)}t_{i}^{(2)} \\
&\quad -0.5(x_{i}^{(5)}-0.1)t_{i}^{(1)}t_{i}^{(3)}
+0.3(x_{i}^{(6)}+0.1)t_{i}^{(2)}t_{i}^{(3)} \\
&\quad +0.7(x_{i}^{(7)}+1)t_{i}^{(1)}t_{i}^{(2)}t_{i}^{(3)} \Big\} + 2,
\end{align*}
where vector $\boldsymbol{w}_{x}$ is the weight vector whose elements are independently drawn from the uniform distribution $U(-1,1)$, and the parameter $l \in \{0,1\}$ determines whether interaction effect terms are included ~($l=1$) or excluded ~($l=0$).
We define the case where $l=1$ as the first scenario and the case where $l=0$ as the second scenario.

\section{Implementation Details}
\label{appendix:implement_details}

This section provides the implementation details necessary to reproduce the simulation results reported in Section~\ref{main:simulation}.
We first specify the network architectures of the proposed CIHSI-Net, including representation learning, task embedding, and outcome prediction subnetworks.
Subsequently, we outline the optimization protocols, hyperparameter settings, and the implementation configurations for the baseline methods.

\paragraph{CIHSI-Net Architecture and Training.}
All neural networks within CIHSI-Net are constructed using fully connected~(FC) layers~\citep{lecun2015deep} with Leaky ReLU activation functions~\citep{xu2020reluplex}.
The specific architectures are as follows.
\begin{itemize}
\item \textbf{Representation Learning Network $\phi$}: Consists of two hidden layers with 100 units each.
The output dimension is set to 64.
Crucially, we apply batch normalization with a fixed scale to the output of the representation learning network because the BFG-WB regularization is sensitive to feature scaling, similar to the balancing regularization in \citet{shalit2017estimating}.
\item \textbf{Task Embedding Network $t_w$}: Consists of two hidden layers with 100 units each, outputting a five-dimensional embedding vector.
\item \textbf{Outcome Prediction Network $h$}: Consists of two hidden layers with 100 units each.
It takes the concatenated output of $\phi$ and $t_w$ as input.
\end{itemize}

During training, the Wasserstein barycenter is estimated for each mini-batch as a free-support $W_2^2$ barycenter with 16 uniformly weighted support points, using the treatment-pattern-specific empirical representation distributions. 
In all experiments, the weights $\lambda_{\boldsymbol{t}}$ in Equation~\eqref{eq:def_barycenter} are set uniformly across the treatment patterns included in the barycenter computation.
The BFG-WB weights are also set uniformly as $w_{\boldsymbol{t}}=2^{-K}$ for all treatment patterns.
We compute the barycenter through iterative updates with a maximum of 100 iterations.
The optimization procedure of CIHSI-Net is encapsulated in Algorithm~\ref{alg:cihsi_net_bfgwb}.

\begin{algorithm}[t]
\caption{The computation workflow of CIHSI-Net with BFG-WB}
\label{alg:cihsi_net_bfgwb}

\KwIn{
observed data
$\mathcal{D}
= \{(\boldsymbol{x}_i,\boldsymbol{t}_i,y_i)\}_{i=1}^{N}$;
representation mapping $\phi$;
task embedding mapping $t_w$;
outcome mapping $h$.
}

\KwParam{
$\alpha$: strength of BFG-WB regularization;
$\eta$: FGW trade-off parameter;
$\beta$: strength of $L_2$ regularization;
$\{\lambda_{\boldsymbol{t}}\}_{\boldsymbol{t}\in\{0,1\}^{K}}$:
Wasserstein barycenter weights;
$\{w_{\boldsymbol{t}}\}_{\boldsymbol{t}\in\{0,1\}^{K}}$:
BFG-WB weights.
}

\KwOut{
$\mathcal{L}$: the learning objective of CIHSI-Net.
}

\Indentp{2em}
\SetNlSkip{-1.5em}

$\boldsymbol{R}_i
\leftarrow
\phi(\boldsymbol{x}_i)$\;

$\boldsymbol{R}_{\boldsymbol{t}}
\leftarrow
\{\boldsymbol{R}_i
\mid
\boldsymbol{t}_i=\boldsymbol{t}\}$
for each
$\boldsymbol{t}\in\{0,1\}^{K}$\;

$\boldsymbol{R}_b^{*}
\leftarrow
\arg\min_{\boldsymbol{R}_{b}}
\sum_{\boldsymbol{t}\in\mathcal{T}}
\lambda_{\boldsymbol{t}}
W^{2}_{2}
(\boldsymbol{R}_{\boldsymbol{t}},\boldsymbol{R}_{b})$\;

$\mathcal{L}_{\phi}
\leftarrow
\sum_{\boldsymbol{t}\in\mathcal{T}}
w_{\boldsymbol{t}}
F(\boldsymbol{R}_{\boldsymbol{t}},\boldsymbol{R}_b^{*})$\;

$\widehat{p}(\boldsymbol{t})
\leftarrow
\frac{1}{N}
\sum_{i=1}^{N}
\mathbb{I}[\boldsymbol{t}_i=\boldsymbol{t}]$
for each
$\boldsymbol{t}\in\{0,1\}^{K}$\;

Calculate $\mathcal{L}_{y}$ following
\eqref{eq:weighted_loss}
with $\widehat{p}(\boldsymbol{t})$
and outcome mapping $h$\;

$\mathcal{L}
\leftarrow
\mathcal{L}_y
+
\alpha\mathcal{L}_{\phi}
+
\beta\lVert\boldsymbol{w}\rVert_2^2$\;

\end{algorithm}

\begin{table}[t]
\centering
\begin{tabular}{lcc}
\toprule
Hyperparameter & Candidate values & Selected value \\
\midrule
BFG-WB coefficient $\alpha$
& $\{0.1,0.5,1.0,1.5,2.0,5.0\}$
& $1.0$ \\
FGW trade-off $\eta$
& $\{10^{-5},0.1,0.2,\ldots,0.9,1.0\}$
& $0.6$ \\
Number of units per FC hidden layer
&  $\{20, 50, 100\}$
& $100$ \\
\bottomrule
\end{tabular}
\caption{Hyperparameter search ranges and selected values.}
\label{tab:hyperparameter_search}
\end{table}

All models are trained with Adam optimizer~\citep{kingma2014adam} for a maximum of 300 epochs, with an early stopping patience set to 15 epochs.
Unless otherwise specified, we used a learning rate of $10^{-4}$, a batch size of 256, and an L2 regularization of $10^{-5}$.
Data splitting is performed by random shuffling, allocating 60\% for training, 10\% for validation, and 30\% for testing.
Table~\ref{tab:hyperparameter_search} summarizes the candidate values considered for each hyperparameter.
For CIHSI-Net, the final configuration are selected as the one that achieved the lowest validation loss on Simulation~1.

\paragraph{Baseline Methods Selection and Implementation.}
To enable a comprehensive comparison, we implemented the following three baseline methods.
These methods cover complementary modeling paradigms: TECE-VAE adopts a latent-variable generative approach, NCoRE explicitly models relations among treatment combinations, and CISI-Net employs pairwise Wasserstein balancing across treatment patterns.
In particular, CISI-Net is the closest baseline to CIHSI-Net and enables us to directly evaluate the contribution of replacing pairwise balancing with the proposed barycenter-based BFG-WB objective.
For each model, we set the hyperparameters as follows.
\begin{itemize}
\item \textbf{TECE-VAE: } The latent dimension is set to 25, and the task embedding network has three hidden layers with 200 units and ELU activation~\citep{clevert2015fast}.
The task embedding network produces a five-dimensional embedding vector.

\item \textbf{NCoRE: } Each interaction subnetwork is implemented as two FC layers with 200 units per layer and ReLU activation~\citep{nair2010rectified}.
The representation learning network has two hidden FC layers with 200 units and  outputs a 100-dimensional vector.

\item \textbf{CISI-Net: } CISI-Net has three neural networks, all of which are built using FC layers and leaky ReLU activation. 
Each network consists of two hidden layers with 100 units each.
The task embedding network outputs a five-dimensional embedding vector.
The balancing penalty coefficient $\alpha$ is set to 0.1, and the Integral Probability Metrics~(IPM) used in the penalty is the 1-Wasserstein distance~\citep{sriperumbudur2010non}.
\end{itemize}
All baseline models are tuned and trained under conditions comparable to CIHSI-Net to ensure a fair evaluation.

We excluded other methods to ensure a focused and fair comparison for the following reasons:
\begin{itemize}
    \item Some methods have been empirically shown to underperform the baselines adopted in this study. In particular, methods originally developed for single-treatment settings, such as Treatment-Agnostic Representation Network~(TARNet)~\citep{shalit2017estimating} and Counterfactual Regression~(CFR)~\citep{shalit2017estimating}, have been shown in previous experiments to be outperformed by methods specifically designed for multiple simultaneous treatments~\citep{murakami2025multipletreatmentscausaleffects,saini2019multiple}.

    \item Some methods rely on additional assumptions that are not adopted in our problem setting, such as the sequential ignorability assumption used in Single-cause Perturbation~\citep{qian2021estimating}.

    \item Some methods target different estimands or data structures, such as H-learner, which considers multiple outcomes~\citep{Chauhan_2025}.

    \item Some methods share modeling paradigms and limitations similar to those already represented by the selected baselines, and therefore provide limited additional insight; Variational Sample Re-weighting~\citep{zou2020counterfactual}, for example, represents a latent-variable approach similar to TECE-VAE.

    \item Some methods employ architectures closely related to those of the selected baselines but are not applicable to binary treatments. For example, the Dosage Combination Network~\citep{schweisthal2023reliable} is architecturally similar to CISI-Net but is specifically designed for combinations of continuous treatment dosages and therefore cannot be directly applied to the binary treatment setting considered here.
\end{itemize}
Furthermore, conventional inverse probability weighting based and doubly robust estimators, are not included because previous empirical comparisons have shown them to be outperformed by TARNet or CFR~\citep{shalit2017estimating}.

\paragraph{Hardware and Software Environments}
All experiments are conducted in a single-GPU computing environment running Ubuntu 22.04.5 LTS on an x86\_64 architecture. 
The environment are equipped with one NVIDIA L4 GPU with 23\,GB of GPU memory, 12 virtual CPUs corresponding to an Intel Xeon processor operating at 2.20\,GHz, and 52\,GB of system memory. 
CUDA compatibility version is 13.0.
All methods are implemented in Python 3.12, and the neural network models are implemented primarily using TensorFlow 2.20.0. 

\section{Additional Simulation Results}
\subsection{Comparison of Computing Speed}
\label{appendix:exp_computing_speed}
\begin{table}[tb]
  \centering
  \renewcommand{\arraystretch}{1.1}  
  \begin{tabular}{lcccccc}
    \toprule
    Method & $K=2$ & $K=3$ & $K=4$ & $K=5$ & $K=6$ & $K=8$ \\
    \midrule
    CISI-Net  & 5.98 & 15.08 & 31.09 & 138.54 &  125.62 & 1529.83 \\
    CIHSI-Net & 8.21 &  9.02 & 13.92 &  22.71 &  23.01 &  48.32 \\
    \bottomrule
  \end{tabular}
  \caption{
  Wall-clock training time per epoch~(s/epoch) for CISI-Net~(pairwise balancing) and CIHSI-Net~(BFG-WB), evaluated with $K \in \{2,3,4,5,6,8\}$ treatments under the same hardware and training configuration. 
  For each $K$, both methods are trained with the same batch size, while the batch size is adjusted across different values of $K$.
  Times are averaged across all observed epochs in 100 trials with early stopping~(max 300 epochs, patience = 15).}
  \label{tab:compute_cost}
\end{table}
One advantage of implementing BFG-WB within CIHSI-Net is its lower computational cost relative to existing pairwise balancing methods.
OT-based discrepancy evaluation, as used in representation balancing, requires solving optimization problems. 
This requirement often introduces a significant bottleneck during model training~\citep{genevay2018learning,yamada2022approximating1wassersteindistancetrees}.
If the number of treatments is denoted as $K$, the total number of possible treatment patterns is $L = 2^K$.
The number of OT computations required for discrepancy evaluation in pairwise balancing increases quadratically with the number of treatment patterns, specifically following the combination formula ${2^K \choose 2}=O(L^2)$.
In contrast, BFG-WB reduces the required number of OT computations to $O(L)$ and demonstrates a linear increase in computation count relative to the number of treatment patterns.
To determine whether this theoretical reduction in computation translates into a measurable decrease in actual training time, we compare the required training time.

We compare computational costs using the CISI-Net and CIHSI-Net with simulation dataset 1~(detailed in Appendix~\ref{appendix:dgp}).
The number of treatments K is set to $K \in \{2,3,4,5,6,8\}$ to verify the scaling of training time as the number of treatment patterns increases to $2^K$.
For each setting, we conduct 100 independent trials using different random seeds for data generation.
Each trial follows the hyperparameters and training protocol described in Appendix~\ref{appendix:implement_details}.
To ensure that each mini-batch contains at least four samples from each treatment pattern, we adjust the batch size according to $K$: 256 for $K \in \{2,3,4,5\}$, 1024 for $K=6$, and 2048 for $K=8$.
We compute the training time per epoch~(s/epoch) from the wall-clock measurements and compare the actual training time required for each method by averaging over all epochs across 100 trials.
Table~\ref{tab:compute_cost} summarizes the results.

Table~\ref{tab:compute_cost} demonstrates  that the computational advantage of CIHSI-Net becomes clear as the number of treatments increases.
For small treatment settings such as $K=2$, this advantage is limited, and CIHSI-Net is slightly slower than CISI-Net due to the overhead of computing the barycenter-based discrepancy.
However, from $K=3$ onward, CIHSI-Net consistently achieves shorter training times than CISI-Net.
The gap generally becomes more pronounced as $K$ increases.
For example, when $K=5$, CISI-Net requires 138.54~[s/epoch], whereas CIHSI-Net requires only 22.71~[s/epoch].
Although the wall-clock time of CISI-Net decreases slightly from $K=5$ to $K=6$ because the batch size increases from 256 to 1024, thereby reducing the number of mini-batch updates per epoch, CIHSI-Net remains substantially faster in both settings.
At $K=8$, this computational advantage becomes particularly pronounced: CISI-Net requires 1529.83~[s/epoch], whereas CIHSI-Net requires only 48.32~[s/epoch].
These results suggest that CIHSI-Net offers practical advantages in training time, particularly in settings with an increased number of treatments.

\subsection{Scalability to a Larger Number of Treatments}
\label{appendix:simulation_scalability}
We investigate the scalability of the proposed method by evaluating whether its estimation accuracy remains superior to that of existing multiple-treatment methods as the treatment dimension increases.
Multiple-treatment settings are practically relevant because many real-world studies involve the simultaneous evaluation of multiple treatments, which requires the estimation of numerous causal effects~\citep{parmar2008speeding, freidlin2008multi}.
Therefore, we compare the estimation accuracy of each method under different values of $K$.

In this scalability analysis, we set the number of treatments to $K \in \{2,4,6,8\}$.
This range is chosen to reflect realistic multiple-treatment settings, since empirical applications rarely involve more than ten simultaneous treatments~\citep{cho2022aging, freidlin2008multi, muralidharan2025factorial, Tsuboi31122024}.
Therefore, setting the maximum number of treatments to $K=8$ provides a practically meaningful evaluation range while still requiring the methods to handle a rapidly increasing number of treatment patterns.
Across all settings, we use a common functional form for treatment assignment and outcome generation, and fix the sample size at $N=50{,}000$.
The covariates, treatment assignments, and outcomes are generated according to the following data-generating process:
\begin{gather*}
x_{i,n}^{(j)} \sim N(c_n^{(j)}, 1^2), \quad 
x_{i,u}^{(j)} \sim U(-1,1), \quad 
\boldsymbol{x}_{i} = (\boldsymbol{x}_{i,n}, \boldsymbol{x}_{i,u}),\\
H_{i}^{(v)} = \mathbb{I}\left(\,x_{i}^{(v)} + x_{i}^{(v+1)} > 1\,\right), \quad 
\boldsymbol{H}_{i} = (H_{i}^{(1)}, H_{i}^{(2)}, H_{i}^{(3)}),\\
t_{i}^{(k)} \sim \mathrm{Bern} \left( \sigma \left( \boldsymbol{w}_{t_k}^\top \boldsymbol{x}_{i}- \boldsymbol{w}_{H}^{\top}\boldsymbol{H}_{i} \right) \right),\quad
\boldsymbol{t}_{i} = (t_{i}^{(1)},\dots,t_{i}^{(K)}),\\
Y_i \sim \mathcal{N}\!\left(g\!\left(\boldsymbol{x}_{i},\boldsymbol{t}_i\right),\,0.1^2\right),
\end{gather*}
where $j \in \{1,\dots,15\}$ denotes the covariate index, $v \in \{1,2,3\}$, and $k \in \{1,\dots,K\}$.
$\mathbb{I}(\cdot)$ is the indicator function, $\sigma(\cdot)$ is the sigmoid function defined as $\sigma(x) = 1 / (1 + \exp(-x))$.
The elements of $\boldsymbol{w}_{t_k}$ are independently drawn from $U(-0.5,0.5)$, and the elements of $\boldsymbol{w}_{H}$ are independently drawn from $U(-1,1)$.
The outcome generating function $g$ is defined as follows:
\begin{gather*}
g(\boldsymbol{x}_{i},\boldsymbol{t}_i)
=
\boldsymbol{w}_{x}^{\top}\boldsymbol{x}_{i}
+
\sum_{\substack{\emptyset \neq S \subseteq \{1,\dots,K\} \\ |S| \le K}}
\tau_S(\boldsymbol{x}_i)\prod_{k\in S} t_{i}^{(k)}
+2, \ \ 
\tau_S(\boldsymbol{x}_i) = w_{S}  \left(1 + x_i^{(j_{S})}\right),
\end{gather*}
where $j_{S} \in \{1,\dots,30\}$ denotes a covariate index randomly assigned to each subset $S$.
The vector $\boldsymbol{w}_{x}$ is a weight vector whose elements are independently drawn from the uniform distribution $U(-1,1)$, and each $w_S$ is independently drawn from $U(0.2,1.0)$ with a randomly assigned sign.
This setup induces treatment assignment bias through $\boldsymbol{H}_i$ and heterogeneous treatment effects through $\tau_S(\boldsymbol{x}_i)$.

All model hyperparameters are kept identical to those described in Appendix~\ref{appendix:implement_details}.
The only exception is the batch size, which is adjusted depending on $K$.
Specifically, to ensure that each batch contains at least four samples from each treatment pattern, we set the batch size to 256 for $K \in \{2,4\}$, 1024 for $K=6$, and 2048 for $K=8$.
This adjustment is necessary to obtain stable empirical estimates of the balancing regularization terms, because the number of treatment patterns increases as $2^K$.
We exclude TECE-VAE from this scalability analysis because it does not show competitive performance in the main simulation results reported in Section~\ref{main:simulation}.

\begin{figure}[tb]
  \centering
  \includegraphics[width=\linewidth, keepaspectratio]{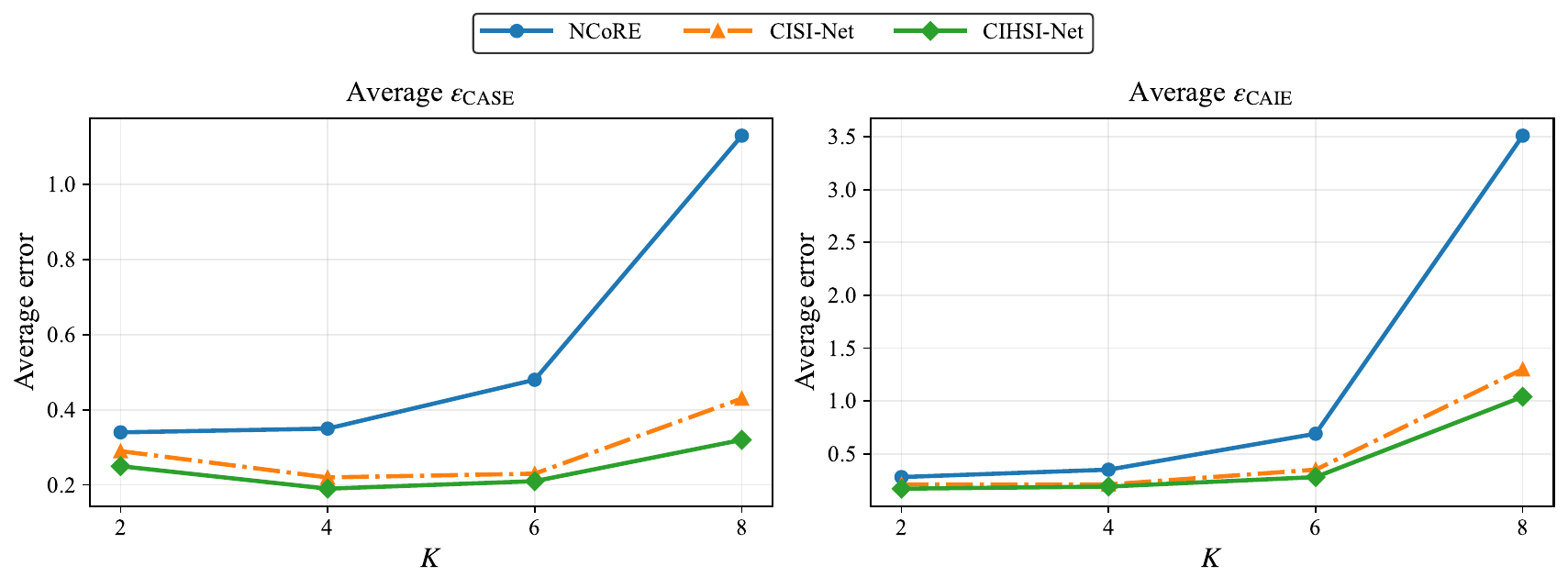}
  \caption{Average $\epsilon_{\mathrm{CASE}}$~(left) and $\epsilon_{\mathrm{CAIE}}$~(right) under varying numbers of treatments $K \in \{2,4,6,8\}$.}
  \label{fig:ase_aie_errors_treatmentNum}
\end{figure}

Figure~\ref{fig:ase_aie_errors_treatmentNum} shows the average $\epsilon_{\mathrm{CASE}}$ and $\epsilon_{\mathrm{CAIE}}$ over 50 trials across varying numbers of treatments $K \in \{2,4,6,8\}$.
CIHSI-Net consistently achieves the lowest estimation errors for both CASE and CAIE across all values of $K$, which demonstrates that the proposed method maintains high estimation accuracy even as $K$ increases.
Crucially, the performance gap between CIHSI-Net and CISI-Net becomes larger as $K$ increases, and CIHSI-Net achieves this higher accuracy with substantially shorter training time than CISI-Net, as shown in Table~\ref{tab:compute_cost}.
This tendency suggests that pairwise balancing becomes unstable as the number of treatment patterns grows, whereas the barycenter-based alignment in CIHSI-Net provides a more consistent global reference across treatment patterns.
These results demonstrate that CIHSI-Net improves both estimation accuracy and computational efficiency in settings with many treatment patterns, supporting its scalability for multiple-treatment causal inference.
\subsection{Detailed Results of Ablation Study}
\label{appendix:ablation_results_details}
In this section, we provide the detailed values for $\epsilon_{\mathrm{CASE}}$ and $\epsilon_{\mathrm{CAIE}}$.
Although the average results in Section~\ref{main:simulation} clarify the overall tendency, the detailed results further confirm whether the observed improvement is consistent across individual CASE and CAIE estimation errors.
Table~\ref{tab:simulation_ablation_detail} shows the detailed ablation results on simulation dataset 1.

The detailed results also show that the combination of the Wasserstein barycenter and the FGW discrepancy achieves the lowest estimation errors across all reported CASE and CAIE errors.
These results are consistent with the discussion in Section~\ref{main:simulation} and provide additional support for the design of BFG-WB.

\begin{table*}[tb]
    \centering
    \small
    \setlength{\tabcolsep}{6.8pt}
    \begin{tabular}{lccccc|cccc}
    \hline
    \multicolumn{3}{l}{} & \multicolumn{3}{c|}{${\epsilon_{\mathrm{CASE}}}$} & \multicolumn{4}{c}{${\epsilon_{\mathrm{CAIE}}}$} \\
    \cline{4-6}\cline{7-10}
    No. & WB & Penalty type & $k=1$ & $k=2$ & $k=3$ & $S=\{1,2\}$ & $S=\{2,3\}$ & $S=\{1,3\}$ & $S=\{1,2,3\}$ \\
    \hline
    1 & \xmarkgray & None & 0.28 $\pm$ 0.08 & 0.29 $\pm$ 0.04 & 0.25 $\pm$ 0.09 & 0.36 $\pm$ 0.03 & 0.12 $\pm$ 0.04 & 0.16 $\pm$ 0.03 & 0.31 $\pm$ 0.04 \\
    2 & \xmarkgray & W & 0.22 $\pm$ 0.04 & 0.24 $\pm$ 0.03 & 0.19 $\pm$ 0.06 & 0.31 $\pm$ 0.07 & 0.11 $\pm$ 0.04 & 0.15 $\pm$ 0.08 & 0.31 $\pm$ 0.08 \\
    3 & \xmarkgray & GW     & 0.25 $\pm$ 0.03 & 0.28 $\pm$ 0.04 & 0.23 $\pm$ 0.10 & 0.35 $\pm$ 0.03 & 0.12 $\pm$ 0.06 & 0.16 $\pm$ 0.04 & 0.32 $\pm$ 0.05 \\
    4 & \xmarkgray & FGW & 0.23 $\pm$ 0.04 & 0.25 $\pm$ 0.03 & 0.19 $\pm$ 0.05 & 0.33 $\pm$ 0.04 & 0.12 $\pm$ 0.06 & 0.16 $\pm$ 0.04 & 0.31 $\pm$ 0.04 \\
    \midrule
    5 & \cmark     & W & 0.22 $\pm$ 0.03 & 0.24 $\pm$ 0.03 & 0.19 $\pm$ 0.08 & 0.30 $\pm$ 0.04 & 0.10 $\pm$ 0.05 & 0.15 $\pm$ 0.09 & 0.28 $\pm$ 0.07 \\
    6 & \cmark     & GW & 0.22 $\pm$ 0.04 & 0.24 $\pm$ 0.03 & 0.19 $\pm$ 0.08 & 0.32 $\pm$ 0.06 & 0.10 $\pm$ 0.04 & 0.16 $\pm$ 0.05 & 0.29 $\pm$ 0.06 \\
    7 & \cmark     & FGW   & \textbf{0.19 $\pm$ 0.04} & \textbf{0.21 $\pm$ 0.04} & \textbf{0.18 $\pm$ 0.10} & \textbf{0.28 $\pm$ 0.05} & \textbf{0.08 $\pm$ 0.05} & \textbf{0.12 $\pm$ 0.04} & \textbf{0.24 $\pm$ 0.04} \\
    \hline
    \end{tabular}
    \caption{
    Ablation study result on simulation dataset 1.
    WB indicates whether the Wasserstein barycenter is used.
    Penalties: Wasserstein~(W), Gromov-Wasserstein~(GW), Fused GW~(FGW), and None.
    }
    \label{tab:simulation_ablation_detail}
\end{table*}
\subsection{Analysis of Representation Discrepancy}
\label{appendix:learned_rep}

We investigate the effectiveness of the proposed method in reducing distributional discrepancies across treatment patterns compared to existing approaches.
Minimizing these discrepancies serves two critical objectives essential for accurate causal inference: (i) reducing selection bias arising from treatment assignment imbalances, which corresponds to minimizing distributional discrepancy~(assessed via Wasserstein distance), and (ii) preserving local proximity structures to stabilize the estimation of heterogeneous causal effects, which corresponds to maintaining geometric consistency~(assessed via Gromov-Wasserstein distance).

To empirically verify these properties, we compute these two metrics~(the Wasserstein distance and the Gromov-Wasserstein distance) for every unique pair of treatment patterns on the dataset of Simulation 1 described in Section~\ref{main:simulation}.
We compare three methods: CISI-Net~(without balancing), CISI-Net~(with pairwise balancing), and our proposed CIHSI-Net~(with BFG-WB).

Figure~\ref{fig:balancing_rep} displays the boxplots of pairwise Wasserstein and Gromov-Wasserstein distances calculated across all treatment pairs.
As shown in the left panel~(Wasserstein distance), CIHSI-Net consistently achieves the lowest median and variance of discrepancies across representation distributions compared to the baselines.
Whereas standard CISI-Net reduces discrepancies relative to the non-balanced model, it fails to match the compactness achieved by CIHSI-Net.
This result suggests that simply extending pairwise balancing is insufficient for achieving comprehensive alignment in multiple-treatment settings.
The right panel~(Gromov-Wasserstein distance) further illustrates the advantage of our CIHSI-Net.
CIHSI-Net exhibits lower structural discrepancies than CISI-Net.
This result indicates that independent pairwise alignment can distort the local proximity structures within the representation space.
In contrast, by aligning all distributions toward a shared Wasserstein barycenter, CIHSI-Net effectively mitigates both distributional divergence and local structural inconsistency.
These results empirically demonstrate the superiority of CIHSI-Net in handling multiple treatments.

\begin{figure}[tb]
  \centering
  \includegraphics[width=\linewidth, keepaspectratio]{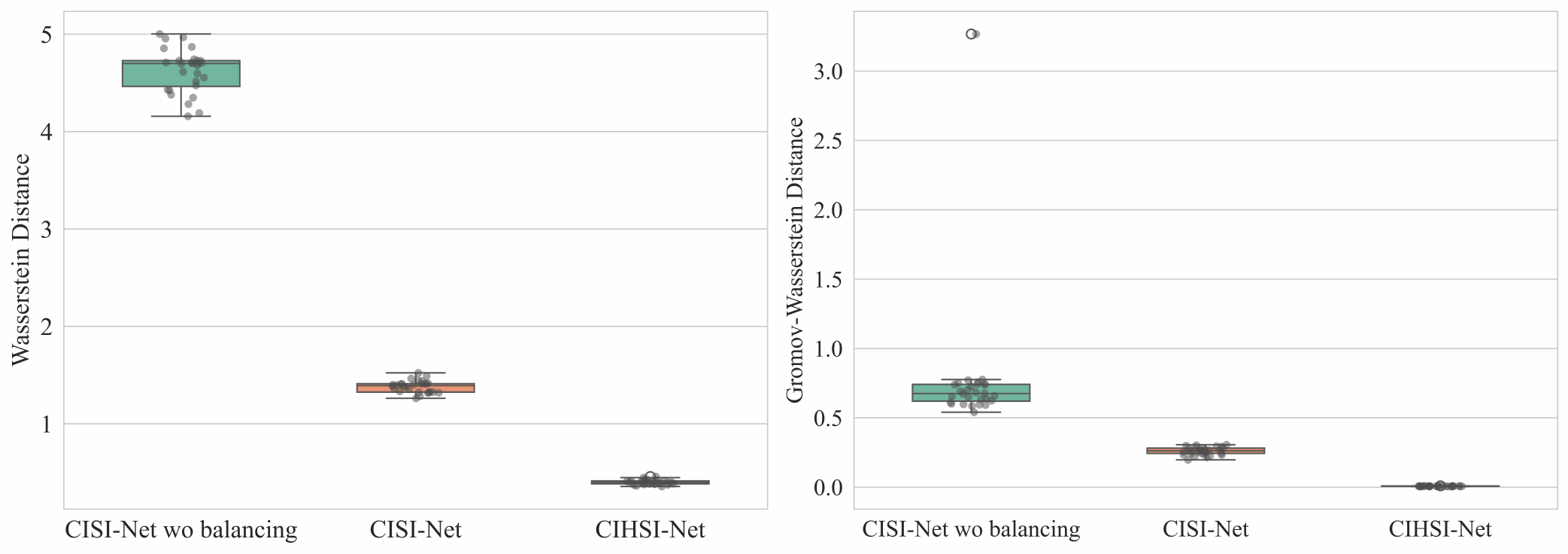}
  \caption{Distribution of pairwise Wasserstein~(left) and Gromov-Wasserstein distances~(right) between representation distributions for all unique treatment pairs. We compare CISI-Net~(without balancing), CISI-Net~(with pairwise balancing), and CIHSI-Net~(ours). Lower values indicate better alignment across treatment patterns.}
  \label{fig:balancing_rep}
\end{figure}

\subsection{Sensitivity Analysis for $\alpha$ and $\eta$}
\label{appendix:sensitivity_analysis}
To examine the sensitivity of CIHSI-Net to the key hyperparameters in BFG-WB, we conduct a sensitivity analysis over the regularization strength $\alpha$ and the FGW trade-off parameter $\eta$.
This analysis clarifies how the balance between global alignment and local structure preservation affects estimation accuracy, and assesses whether CIHSI-Net remains stable across a reasonable range of hyperparameter choices.
Varying $\alpha \in \{0.1,0.5,1.0,1.5,2.0,5.0\}$, we cover weak to strong regularization via BFG-WB, which examines the trade-off between representation balancing and outcome prediction accuracy.
Additionally, to compare the relative contributions of local structure preservation and feature-wise distribution alignment, we vary $ \eta \in \{10^{-5}, 0.1, 0.2, \ldots, 0.9, 1.0\}$. 
Here, larger $\eta$ places more weight on feature-wise Wasserstein alignment, whereas smaller $\eta$ emphasizes the Gromov-Wasserstein-based structural term.
We include $\eta=10^{-5}$ as a near-boundary configuration that approximates GW-only alignment while remaining within the theoretical domain $\eta\in(0,1]$.

Figure~\ref{fig:alpha_eta_anal} summarizes the sensitivity results over $\alpha$ and $\eta$, averaged over 100 trials.
\begin{figure*}[tb]
  \centering
  \includegraphics[width = \linewidth, keepaspectratio]{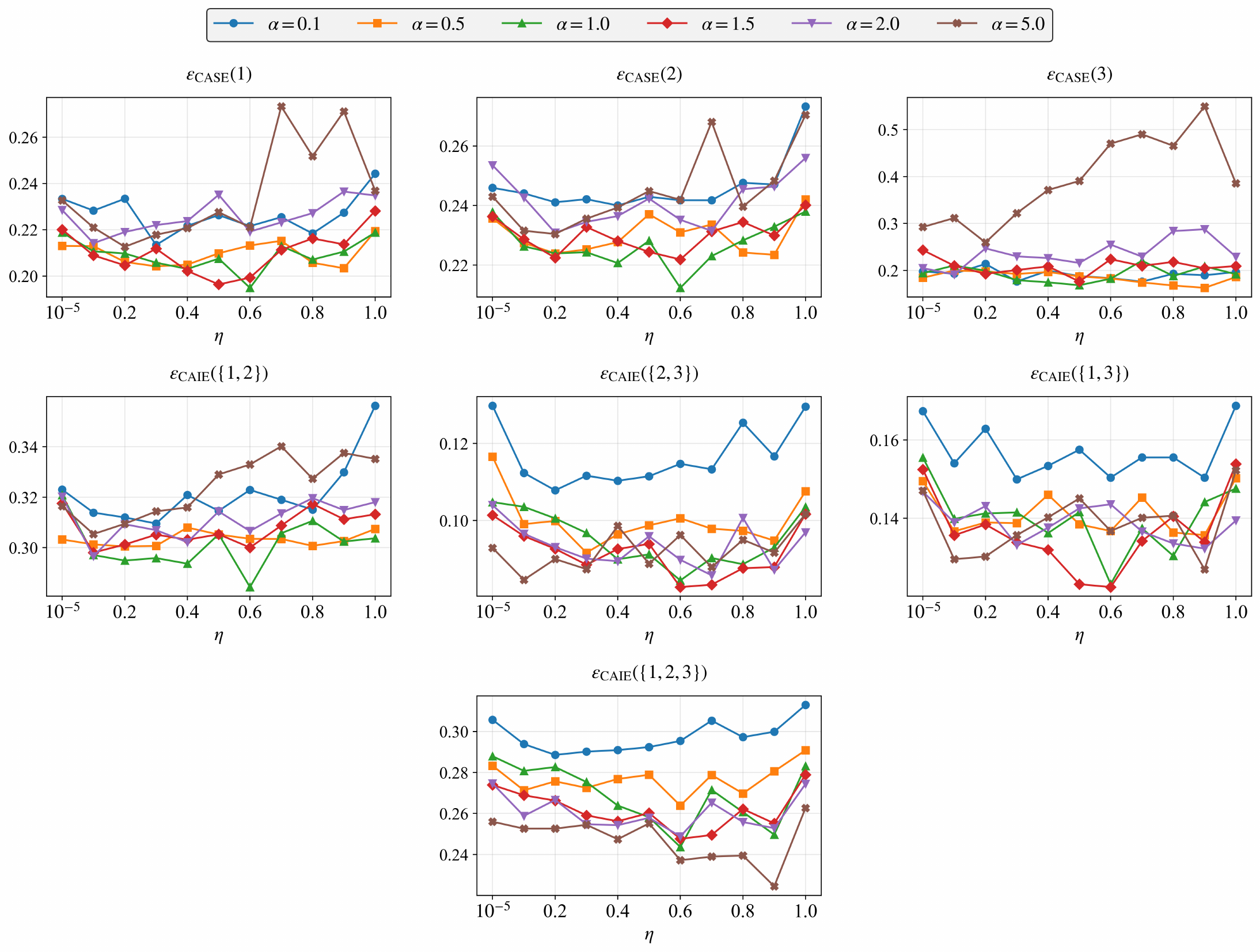}
  \caption{Sensitivity analysis of hyperparameters $\alpha$ and $\eta$ on simulation dataset 1.
  The top row shows $\epsilon_{\mathrm{CASE}}(1)$ - $\epsilon_{\mathrm{CASE}}(3)$, the middle row shows $\epsilon_{\mathrm{CAIE}}(\{1,2\})$ - $\epsilon_{\mathrm{CAIE}}(\{1,3\})$, and the bottom panel shows $\epsilon_{\mathrm{CAIE}}(\{1,2,3\})$.
  }
  \label{fig:alpha_eta_anal}
\end{figure*}
The results indicate that estimation performance depends on both $\alpha$ and $\eta$, which emphasizes the importance of appropriately selecting their combination.
For $0.5 < \eta < 0.8$, ${\epsilon_{\mathrm{CASE}}}$ and ${\epsilon_{\mathrm{CAIE}}}$ remain small, suggesting that a moderate emphasis on feature-based alignment while retaining structural preservation is beneficial for causal effect estimation.
At the near-boundary setting $\eta=10^{-5}$, the estimation errors remain competitive for some estimands but are not consistently minimized across metrics and values of $\alpha$. 
This pattern is consistent with the theoretical role of $\eta$: because the error bounds contain a factor of $1/\eta$, they become increasingly loose as $\eta$ approaches zero and cease to provide finite control at the GW-only boundary $\eta=0$.
With respect to $\alpha$, accurate estimation of both CASE and CAIE is achieved when $1.0 \le \alpha \le 1.5$.
Weak balancing regularization, such as $\alpha=0.1$, tends to result in insufficient correction of selection bias and larger estimation errors across all metrics.
Conversely, at $\alpha=5.0$, CASE estimation errors increase, while CAIE estimation errors are reduced.
This may be interpreted as stronger selection bias affecting units receiving multiple treatments simultaneously, where strong balancing effectively estimated interaction treatment effects.

\subsection{Evaluation on Semi-Synthetic Dataset}
\label{appendix:evaluation_semi_synthetic}
We conduct an additional semi-synthetic experiment to evaluate the proposed method under a more realistic setting while retaining access to ground-truth causal effects.
Whereas the fully synthetic simulations allow us to evaluate the proposed method under controlled data-generating processes, they may not fully reflect the covariate distribution and treatment imbalance observed in real-world data.
The purpose of this experiment is to verify whether CIHSI-Net can accurately estimate both CASE and CAIE when the covariates and treatment assignments follow an empirical distribution, while the outcome-generation mechanism is known.

We construct a semi-synthetic dataset using the real covariates and observed treatment assignments from Dataset A used in Section~\ref{main:application_marketing_promotions}.
Specifically, we pair Dataset A's real covariates and treatment assignments with a known synthetic outcome function that includes nonlinear heterogeneous effects for both CASE and CAIE.
Let $\boldsymbol{x}_i \in \mathbb{R}^{71}$ denote the observed covariates for unit $i$.
The observed treatment vector $\boldsymbol{t}_i=(t_{i}^{(1)},t_{i}^{(2)},t_{i}^{(3)}) \in \{0,1\}^3$ is kept fixed from the original data.
The observed outcome is then generated from the following semi-synthetic data-generating process:
\[
y_i = g(\boldsymbol{x}_i,\boldsymbol{t}_i) + \varepsilon_i, \qquad \varepsilon_i \sim \mathcal{N}(0,0.5^2).
\]
The outcome generating function $g$ is defined as:
\[ 
g(\boldsymbol{x}_i,\boldsymbol{t}) = \tanh(\boldsymbol{x}_i^\top \boldsymbol{W}_g+b_g) + \sum_{\emptyset \neq S \subseteq \{1,2,3\}} g_S(\boldsymbol{x}_i) \prod_{k \in S} t_k, 
\]
where, for each non-empty subset $S \subseteq \{1,2,3\}$, $g_S(\boldsymbol{x}_i) = w_S \tanh(\boldsymbol{x}_i^\top \boldsymbol{W}_S+b_S)$ represents the heterogeneous treatment effect function.
The elements of $\boldsymbol{W}_g$ and $\boldsymbol{W}_S$ are independently drawn from $\mathcal{N}(0,1^2)$, and the bias terms $b_g$ and $b_S$ are drawn from $\mathcal{N}(0,0.5^2)$.
The coefficient $w_S$ is independently drawn from $U(-1,1)$ for each non-empty subset $S$.
This construction preserves the covariate distribution and treatment imbalance of the original data while providing known ground-truth CASE and CAIE values.

We use the same model architectures, hyperparameters, and training protocol as those described in Appendix~\ref{appendix:implement_details}.
The data was randomly split into 60\% training, 10\% validation, and 30\% test sets, with all reported evaluations performed on the test set.
To ensure statistical reliability, we report average metrics over 100 independent runs.

\begin{table*}[tb]
  \centering
  \small
  \setlength{\tabcolsep}{10.0pt}
  \begin{tabular}{lccc|cccc}
  \hline
  & \multicolumn{3}{c|}{${\epsilon_{\mathrm{CASE}}}$} & \multicolumn{4}{c}{${\epsilon_{\mathrm{CAIE}}}$} \\
  \cline{2-4}\cline{5-8}
  Method & $k=1$ & $k=2$ & $k=3$ & $S=\{1,2\}$ & $S=\{2,3\}$ & $S=\{1,3\}$ & $S=\{1,2,3\}$ \\
  \hline
  TECE-VAE  
    & 1.23 $\pm$ 0.09 
    & 1.25 $\pm$ 0.06 
    & 1.28 $\pm$ 0.07 
    & 3.22 $\pm$ 0.46 
    & 1.82 $\pm$ 0.08 
    & 1.62 $\pm$ 0.17 
    & 4.00 $\pm$ 0.34 \\
  NCoRE     
    & 0.98 $\pm$ 0.16 
    & \textbf{0.98 $\pm$ 0.07} 
    & \textbf{0.78 $\pm$ 0.05} 
    & 2.94 $\pm$ 0.48 
    & 1.29 $\pm$ 0.28 
    & \underline{1.37 $\pm$ 0.14} 
    & 5.73 $\pm$ 1.20 \\
  CISI-Net  
    & \textbf{0.79 $\pm$ 0.01} 
    & 1.13 $\pm$ 0.04 
    & \underline{0.96 $\pm$ 0.12} 
    & \underline{1.49 $\pm$ 0.14} 
    & \underline{1.21 $\pm$ 0.21} 
    & \textbf{0.96 $\pm$ 0.02} 
    & \underline{1.49 $\pm$ 0.11} \\
  CIHSI-Net
    & \underline{0.81 $\pm$ 0.02} 
    & \underline{1.02 $\pm$ 0.07} 
    & \textbf{0.78 $\pm$ 0.07} 
    & \textbf{1.45 $\pm$ 0.07} 
    & \textbf{1.09 $\pm$ 0.08} 
    & \textbf{0.96 $\pm$ 0.04} 
    & \textbf{1.45 $\pm$ 0.06} \\
  \hline
  \end{tabular}
  \caption{
  Semi-synthetic experiment results using Dataset A's real covariates and observed treatment assignments.
  Bold values indicate the best performance, and underlined values indicate the second-best performance for each metric.
  }
  \label{tab:semi_synthetic}
\end{table*}

Table~\ref{tab:semi_synthetic} shows the results of the semi-synthetic experiment.
Overall, CIHSI-Net achieves the strongest performance among the baseline methods under the empirical covariate distribution and observed treatment assignments of Dataset A.
Although CISI-Net and NCoRE obtain the lowest error for some individual ${\epsilon_{\mathrm{CASE}}}$, CIHSI-Net remains competitive for CASE and achieves the best results for all ${\epsilon_{\mathrm{CAIE}}}$.
When the errors are averaged over the ${\epsilon_{\mathrm{CASE}}}$ and ${\epsilon_{\mathrm{CAIE}}}$, CIHSI-Net obtains the lowest overall errors, suggesting that the proposed barycenter-based balancing is effective for estimating both single and interaction effects in the semi-synthetic setting.
This result suggests that CIHSI-Net remains robust even when the treatment assignment imbalance is inherited from real-world data.

\section{Real-world Datasets and Additional Results}
\label{appendix:real_world_application}

In this section, we provide supplementary details and additional results for the real-world marketing application.
We first describe the dataset specifications, preprocessing steps, and experimental protocols common to both Dataset A and Dataset B.
Subsequently, we present the analysis results for Dataset B, which involves promotions from competing merchants.

\subsection{Details of Real-World Application}
\label{appendix:sub_real_world_data_details}

\paragraph{Dataset Description and Treatment Assignment.}
The dataset is derived from a mobile payment platform and involves multiple concurrent marketing promotions.
We utilize this proprietary dataset because existing publicly available causal inference benchmarks focus primarily on single-treatment scenarios and lack the complex, simultaneous multiple-treatment interactions observed in real-world marketing environments. 
\begin{itemize}
    \item \textbf{Dataset A~(reported in Section~\ref{main:application_marketing_promotions})}:
    This dataset includes three marketing promotions that were simultaneously conducted: two offline promotions organized by the same merchant group~(denoted as CP\textsubscript{1} and CP\textsubscript{2}) and one online promotion conducted by another merchant group~(denoted as CP\textsubscript{3}).
    Each promotion is represented by a binary treatment indicator, which results in $2^3 = 8$ distinct treatment patterns.
    The observed sample proportions were approximately 2\% for CP\textsubscript{1} only, 15\% for CP\textsubscript{2} only, 42\% for CP\textsubscript{3} only, 0.2–1\% for two-promotion combinations, 0.2\% for all three promotions simultaneously, and the remainder as the control group.
    
    \item \textbf{Dataset B~(reported in this appendix)}:
    This dataset includes two marketing promotions conducted by two different merchants in the same industry~(denoted as CP\textsubscript{4} and CP\textsubscript{5}), yielding $2^2=4$ treatment patterns.
    The observed sample proportions were approximately 29\% for CP\textsubscript{4} only, 4\% for CP\textsubscript{5} only, 1\% for both promotions, and the remainder as the control group.
\end{itemize}
For both datasets, the control group consists of users who were not exposed to any of the promotion during the treatment period and who satisfied the positivity assumption, randomly drawn from users with at least one mobile payment transaction in the month preceding the promotions.
Additionally, the preprocessing and experimental protocol described below are applied identically to both datasets.

\paragraph{Outcome Definition and Preprocessing.}
The outcome variable $Y$ is defined as the total payment amount during the one-month period following the promotion implementation.
To ensure numerical stability and protect data confidentiality, $Y$ was standardized to have zero mean and unit variance.
This preprocessing preserves the sign of the estimated causal effect and ensures interpretational consistency between the standardized scale and the original scale because the standardization is an affine transformation~\citep{thakral2023estimates, murakami2025multipletreatmentscausaleffects}.
The covariates consist of 71 variables, including service usage histories and user demographic attributes.

\paragraph{User Stratification for Heterogeneity Analysis.}
To analyze heterogeneous treatment effects, we stratified users based on their total payment amount in the month before the promotions. 
Users were grouped into intervals of 5,000 JPY (e.g., 0–5,000, 5,000–10,000). 
Users with usage exceeding 50,000 JPY were aggregated into a single high-usage group, resulting in a total of 11 distinct strata. 
This grouping strategy allows us to examine how the sensitivity to marketing incentives varies with prior engagement levels.

\paragraph{Experimental Configuration.}
For the real-world application, we employed the same hyperparameter configuration that achieved the best performance in the simulation experiments~(Section~\ref{main:simulation}). 
This approach was chosen to evaluate the method's robust performance under a reproducible setting without dataset-specific tuning.
The dataset was randomly split into 60\% training, 10\% validation, and 30\% test sets, with all reported evaluations performed on the test set.

\subsection{Sensitivity and Uncertainty Analysis on Dataset A}
\label{appendix:sensitivity_and_uncertainty_analysis_on_datasetA}
\begin{figure*}[p]
    \centering
    \includegraphics[width=\textwidth]{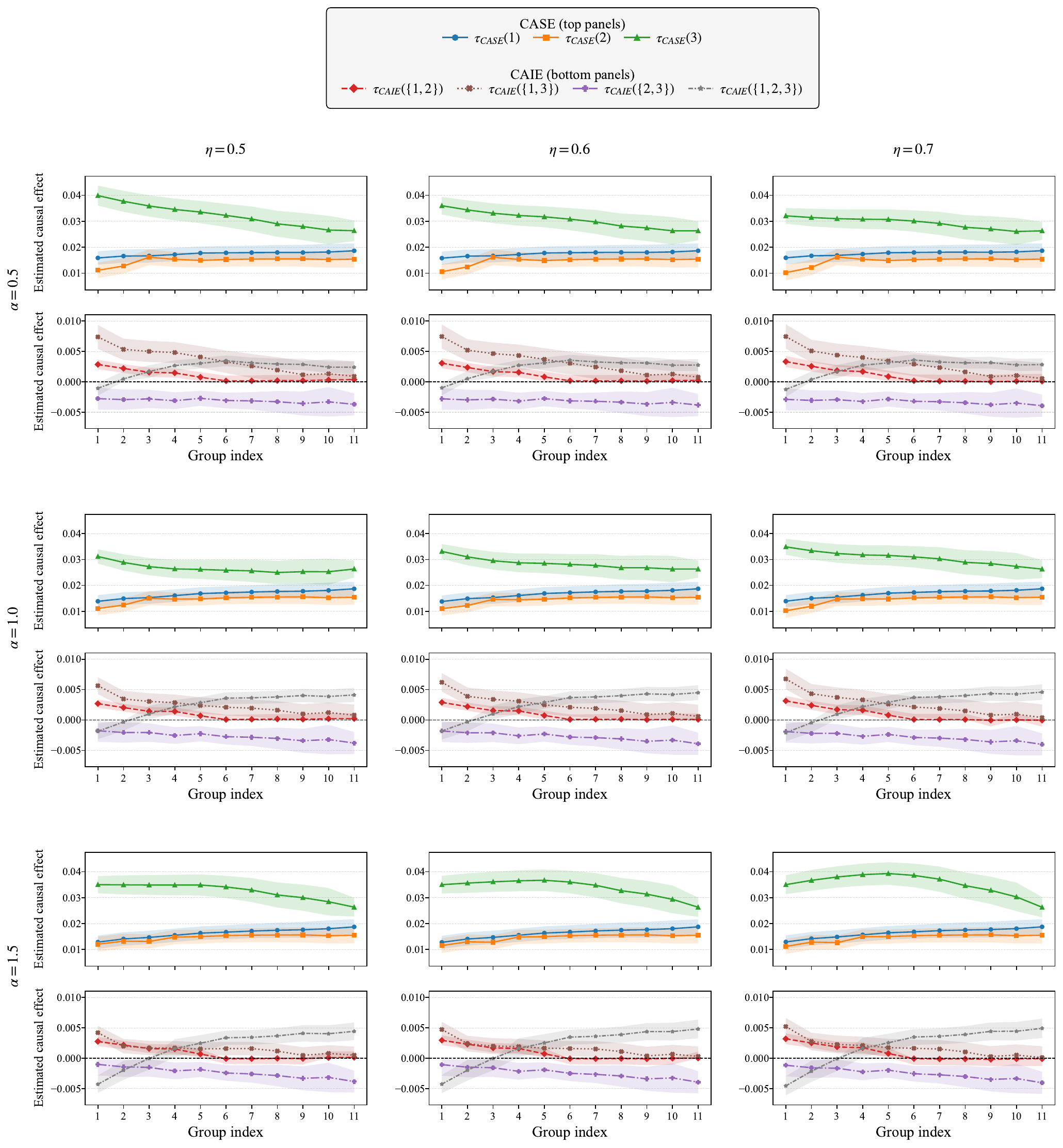}
    \caption{
    Sensitivity and uncertainty analysis of the estimated causal effects in the real-world application on Dataset A.
    The rows correspond to the BFG-WB regularization strength $\alpha \in \{0.5,1.0,1.5\}$, and the columns correspond to the FGW trade-off parameter $\eta \in \{0.5,0.6,0.7\}$.
    In each panel, the top subpanel reports the estimated CASE and the bottom subpanel reports the estimated CAIE across user groups.
    The shaded regions represent bootstrap confidence intervals.
    }
    \label{fig:seisitivity_uncertainty_to_application}
\end{figure*}

This appendix reports additional sensitivity and uncertainty analyses for the estimated causal effects in the real-world application on Dataset A.
In real-world observational data, the ground-truth causal effects are unavailable, and the reliability of the results cannot be assessed through direct estimation errors as in the simulation studies.
Therefore, it is important to examine whether the estimated causal-effect patterns remain stable under different model configurations and sampling variability.
We evaluate the sensitivity of the estimated CASE and CAIE to key BFG-WB hyperparameters and compute bootstrap confidence intervals for the estimated causal effects.
These analyses provide complementary evidence on the robustness of the real-world application.

For the sensitivity analysis, we focus on two key hyperparameters related to BFG-WB: the regularization strength $\alpha$ in Equation~\eqref{eq:total_loss} and the FGW trade-off parameter $\eta$ in Equation~\eqref{eq:fgw_to_barycenter}.
We evaluate nine hyperparameter configurations by varying $\alpha \in \{0.5, 1.0, 1.5\}$ and $\eta \in \{0.5, 0.6, 0.7\}$ while keeping the other model configurations fixed.
For each configuration, the data are randomly split into 60\% training, 10\% validation, and 30\% test sets, and all reported evaluations are performed on the test set.
To quantify uncertainty, we compute bootstrap confidence intervals by repeating the entire estimation pipeline for each hyperparameter configuration~\citep{efron1994introduction,kubota2025impact}.
We repeat this procedure 1,000 times for each hyperparameter configuration.
We then aggregate the test-set estimates for each bootstrap replication and construct confidence intervals from the empirical distribution of the resulting aggregate statistics.
Figure~\ref{fig:seisitivity_uncertainty_to_application} shows the estimated CASE and CAIE with bootstrap confidence intervals under the nine hyperparameter configurations of $\alpha \in \{0.5,1.0,1.5\}$ and $\eta \in \{0.5,0.6,0.7\}$.

Although some individual confidence intervals include zero, the overall treatment-effect patterns remain stable across different hyperparameter settings.
In particular, the estimated effects preserve the same qualitative trends observed in the main real-world analysis, such as the sign reversal of the three-way interaction effect from negative for low-usage users to positive for high-usage users.
This suggests that the empirical findings are not driven by a specific choice of the BFG-WB regularization strength or the FGW trade-off parameter.
Moreover, the consistency of the estimated CASE and CAIE patterns across the bootstrap replications suggests that CIHSI-Net captures the primary interaction structures in the real-world application.
These results provide additional evidence that the real-world findings are robust to both hyperparameter variation and sampling variability.

\subsection{Additional Results on Dataset B}
\label{appendix:additional_real_world_results}
\begin{figure*}[t]
  \centering
  \includegraphics[width = \linewidth, keepaspectratio]{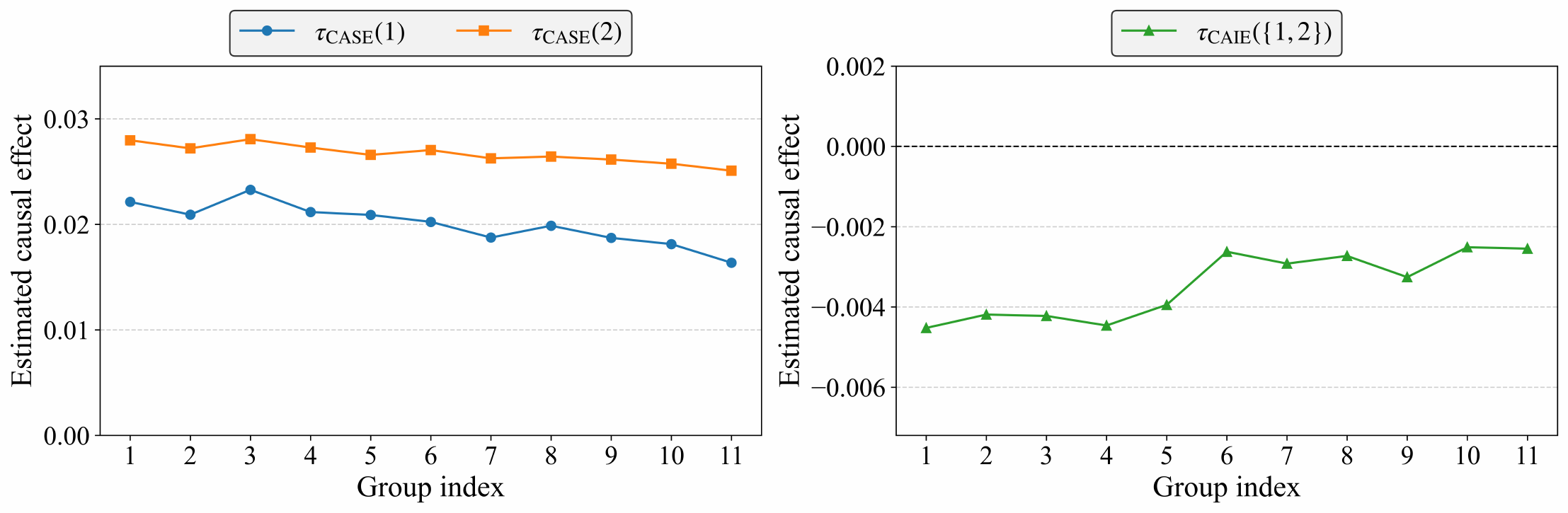}
  \caption{
  In dataset B, estimated CASEs for CP\textsubscript{4} and CP\textsubscript{5} (left) and the estimated CAIE for $\{ \mathrm{CP}\textsubscript{4}, \mathrm{CP}\textsubscript{5} \}$ (right) across user groups stratified by pre-promotion service usage.
  The outcomes are standardized.
  }
  \label{fig:marketing_result_datasetb}
\end{figure*}

Figure~\ref{fig:marketing_result_datasetb} shows the estimated causal effects of dataset B on a standardized outcome scale.
The left panel in Figure~\ref{fig:marketing_result_datasetb} shows that both estimated single-treatment effects are positive across all user groups.
The estimated effect of CP\textsubscript{5} is consistently larger than that of CP\textsubscript{4}, and for both promotions, the effect sizes gradually decrease as the pre-promotion usage level increases. 
This pattern suggests that CP\textsubscript{4} and CP\textsubscript{5} are effective at inducing behavioral changes among low-usage users, whereas their incremental effects are limited for high-usage users whose service usage is already established. 
This heterogeneity is consistent with established findings in marketing research~\citep[e.g.][]{breugelmans2017effect, kubota2025causal, liu2007long}, which suggests that light users typically possess a greater capacity to increase their transaction volume, whereas heavy users are often constrained by saturation in their activity levels.
Collectively, the contrast with Dataset~A~(see Section~\ref{main:application_marketing_promotions}) highlights that treatment effect heterogeneity is highly context-dependent, varying significantly based on promotion characteristics and competitive dynamics.

The right panel in Figure~\ref{fig:marketing_result_datasetb} shows that the estimated interaction effect $\tau_{\mathrm{CAIE}}(\{4,5\})$ is consistently negative across all groups, which suggests that promotions offered by competing merchants in the same industry may induce a reallocation of spending across merchants~(cannibalization). 
The results of this cannibalization are consistent with previous research~\citep{dorotic2021synergistic}.
In contrast to the CASE results, the absolute magnitude of the negative interaction becomes smaller for higher-usage groups, indicating that the degree of cannibalization is mitigated as usage increases. 
This trend suggests that cannibalization may be weaker among high-usage users because their behavior of using multiple merchants is already established.
Overall, these additional results suggest that CIHSI-Net serves as a practical analysis tool by flexibly uncovering heterogeneous single and interaction effects across user groups.

\end{document}